\documentclass[12pt,letterpaper]{article}

\usepackage{natbib,hyperref}
\usepackage[ left=1in, top=1in, right=1in, bottom=1in]{geometry}
\usepackage{graphicx,bm,colonequals,amsmath,amssymb,url,xcolor,bbm}
\usepackage{array,tabularx,multirow}
\usepackage[font={footnotesize}]{caption,subcaption}
\usepackage[utf8]{inputenc}
\usepackage{enumitem}
\usepackage{soul} 
\usepackage{placeins} 
\usepackage[normalem]{ulem}
\usepackage[title]{appendix}

\graphicspath{{plots/}}

\usepackage{mathtools}

\bibpunct[, ]{(}{)}{;}{a}{,}{,}
   
\usepackage{amsthm}
\newtheoremstyle{propstyle} 
    {2mm}                    
    {1mm}                    
    {\itshape}                   
    {}                           
    {\scshape}                   
    {.}                          
    {.5em}                       
    {}  
\theoremstyle{propstyle}

\theoremstyle{propstyle}
\newtheorem{corollary}{Corollary}
\theoremstyle{propstyle}
\newtheorem{proposition}{Proposition}
\theoremstyle{propstyle}

\theoremstyle{propstyle}

\theoremstyle{propstyle}

\theoremstyle{propstyle}

\usepackage[ruled,vlined,noend]{algorithm2e}
\SetKwInput{KwInput}{Input}
\usepackage{algorithmic}

\usepackage{array,framed}
\usepackage{float}
\usepackage{multirow}

\makeatletter
\renewcommand{\paragraph}{%
  \@startsection{paragraph}{4}%
  {\z@}{2ex \@plus 1ex \@minus .2ex}{-1em}%
  {\normalfont\normalsize\bfseries}%
}
\makeatother

\DeclareMathAlphabet\mathbfcal{OMS}{cmsy}{b}{n}

\newcommand{\bd}{\mathbf{d}}
\newcommand{\bb}{\mathbf{b}}

\newcommand{\bx}{\mathbf{x}}
\newcommand{\by}{\mathbf{y}}

\newcommand{\bg}{\mathbf{g}}

\newcommand{\br}{\mathbf{r}}
\newcommand{\bw}{\mathbf{w}}

\newcommand{\bI}{\mathbf{I}}
\newcommand{\bD}{\mathbf{D}}
\newcommand{\bH}{\mathbf{H}}

\newcommand{\bX}{\mathbf{X}}

\newcommand{\bC}{\mathbf{C}}
\newcommand{\bM}{\mathbf{M}}

\newcommand{\bfzero}{\mathbf{0}}

\newcommand{\bftheta}{\bm{\theta}}

\newcommand{\bfSigma}{\bm{\Sigma}}

\newcommand{\cov}{\textbf{Cov}}

\newcommand{\GP}{\mathcal{GP}}

\newcommand{\order}{\mathcal{O}}
\newcommand{\normal}{\mathcal{N}}

\newcommand{\domain}{\mathcal{X}}

\newcommand{\dens}{p}

\title{Scalable Gaussian-process regression and variable selection using Vecchia approximations}

\author{Jian Cao\thanks{Department of Statistics and Institute of Data Science, Texas A\&M University}, Joseph Guinness\thanks{Department of Statistics, Cornell University}, Marc G. Genton\thanks{Statistics Program, King Abdullah University of Science and Technology}, Matthias Katzfuss\footnotemark[1]$^{\ }$\thanks{corresponding author, email: katzfuss@gmail.com}}

\date{}
\begin{document}

\maketitle

\begin{abstract} \noindent
Gaussian process (GP) regression is a flexible, nonparametric approach to regression that naturally quantifies uncertainty. In many applications, the number of responses and covariates are both large, and a goal is to select covariates that are related to the response. For this setting, we propose a novel, scalable algorithm, coined VGPR, which optimizes a penalized GP log-likelihood based on the Vecchia GP approximation, an ordered conditional approximation from spatial statistics that implies a sparse Cholesky factor of the precision matrix. We traverse the regularization path from strong to weak penalization, sequentially adding candidate covariates based on the gradient of the log-likelihood and deselecting irrelevant covariates via a new quadratic constrained coordinate descent algorithm. We propose Vecchia-based mini-batch subsampling, which provides unbiased gradient estimators. The resulting procedure is scalable to millions of responses and thousands of covariates. Theoretical analysis and numerical studies demonstrate the improved scalability and accuracy relative to existing methods.
\end{abstract}


\section{Introduction \label{sec:intro}}

\paragraph{Gaussian process regression}
Many tasks in statistics and machine learning can be viewed as regression problems, with the goal of inferring the functional relationship between a response and a number of covariates. Gaussian processes (GPs) are an attractive choice for modeling the regression function \citep[e.g.,][]{Rasmussen2006}, as they naturally quantify uncertainty, they can flexibly capture nonlinear and nonparametric behavior, they are interpretable, and much of the resulting inference involves closed-form expressions. We focus on GP regression for datasets with a large number of responses, $n$, and a large number of covariates, $d$, under the assumption that only few covariates, $d_0 \ll d$, are useful for predicting the response. In this setting, our goals are variable selection, model estimation, and subsequent prediction based on the selected sparse model.

\paragraph{Existing approaches for large $n$}
Basic GP regression scales poorly to large $n$ or $d$. Many approaches have been proposed that deal with one or both of these issues. The challenge with large $n$ is that direct GP inference requires $\order(n^3)$ time. \citet{Heaton2017} and \citet{Liu2018} provide reviews of methods that tackle the large-$n$ problem in spatial statistics and machine learning, respectively. These methods include fully \citep[e.g.,][]{Quinonero-Candela2005,Banerjee2008,Finley2009} and partially \citep[e.g.,][]{Snelson2007,Sang2011a} independent conditional (FIC/PIC) approximations, but these low-rank approaches can have limitations in many settings \citep[e.g.,][]{Stein2013a}, even when optimizing over pseudo-inputs \citep{hensman2015scalable}. Other GP approximations, such as multi-level PIC \citep{Katzfuss2015,Katzfuss2017b}, approximations based on stochastic partial differential equations \citep{Lindgren2011a}, distributed GPs \citep{deisenroth2015distributed} or KISS-GP \citep{wilson2015kernel}, can struggle with high input dimension $d$.

\paragraph{The Vecchia approximation}
A highly promising approach to scaling GP inference to large $n$ may be the Vecchia approximation \citep{Vecchia1988}, which has become very popular in spatial statistics \citep[e.g.,][]{Stein2004,Datta2016,Guinness2016a,Katzfuss2017a,Katzfuss2018}, but which has not received much attention in machine learning. This approach can be viewed as an ordered conditional approximation, in which the joint density of the GP response is approximated as a product of univariate conditional distributions. The resulting approximation can be highly accurate even with small conditioning sets. \citet{Katzfuss2020} proposed a scaled Vecchia approximation that further improves the accuracy of the Vecchia approximation and used it for GP emulation of expensive computer experiments in $d = \order(10)$ dimensions. A more detailed review of Vecchia approximations will be provided in Section \ref{subsec:review_Vecchia}.

\paragraph{Existing approaches for large $d$}
There has also been extensive work on scaling GPs to moderate or high input dimension $d$. Moderate $d$ can be handled by variable selection using automatic relevance determination (ARD) kernel functions \citep{neal1996bayesian} and Bayesian model selection \citep{dearmon2016g, Posch2021}. However, for larger $d$ (say $d \gg 100$), these methods are not sufficiently scalable due to computation and convergence issues caused by the high dimensionality of the parameter space. For such high dimensions, existing approaches include penalized GP regression \citep[e.g.,][]{yi2011penalized}, manifold GP regression \citep[e.g.,][]{Calandra2016}, and hierarchical diagonal sampling \citep[HDS; e.g.,][]{chen2012joint}. However, both penalized GP and manifold GP regressions consider all covariates simultaneously, leading to $\order(d)$ optimization parameters, which may negatively impact model inference in three aspects, namely convergence to local optima, over-fitting, and computational inefficiency. Furthermore, \cite{yi2011penalized} and \cite{Calandra2016} optimized the exact GP likelihood, not scalable with respect to $n$, while HDS assumes that responses are sampled where needed, mainly addressing Bayesian optimization instead of GP regression.

\paragraph{Large numbers of responses and covariates}
Several methods have been proposed to handle large $n$ and $d$ by approximating the GP using FIC and transforming and reducing the dimension of the input domain, such as randomly-projected additive GPs \citep{delbridge2020randomly}, deep kernel learning \citep{wilson2016deep}, and dimension reduction with pseudo-inputs \citep{Snelson2006}. These approaches mainly achieve dimension reduction rather than variable selection. To our knowledge, none of the existing approaches is suitable for our goal of simultaneous variable selection and GP regression for large $n$ and large $d$.

\paragraph{The VGPR algorithm}
Here we propose the VGPR algorithm, for Vecchia GP Regreesion, which is highly scalable in $n$ and $d$. Specifically, to handle large $n$, we extend the scaled Vecchia GP approximation \citep{Katzfuss2020} and propose Vecchia-based mini-batch subsampling, which provides unbiased gradient estimators. To achieve variable selection for large $d$, we consider a penalized Vecchia-GP loglikelihood, and we traverse the regularization path from strong to weak penalization, sequentially adding candidate covariates based on the gradient of the log-likelihood and deselecting irrelevant covariates through a new quadratic constrained coordinate descent algorithm (QCCD). QCCD builds a quadratic approximation of the objective function at each iteration and applies constrained coordinate descent to find the constrained quadratic optimum. Compared with existing GP regression methods such as \citet{yi2011penalized} and \citet{Katzfuss2020}, traversing the regularization path with warm starts effectively avoids local optima while QCCD can reach boundary values, achieving covariate deselection without artificial thresholding. We provide theoretical and numerical evidence for our gradient-based variable selection. The dominant complexity of VGPR depends linearly on the batch size and quadratically on the number of selected covariates (as opposed to the total number of responses or covariates). 

\paragraph{Outline} In Section \ref{sec:review}, we briefly review ARD kernels and the scaled Vecchia approximation. Section~\ref{sec:GPR_var_select} introduces our new VGPR algorithm that involves the QCCD subroutine, the choice of the penalty function, the selection of covariates based on the gradient, and a mini-batch sampling technique specific to the Vecchia approximation. In Section~\ref{sec:sim_study}, we compare VGPR with state-of-the-art GP regressions in terms of posterior inference and variable selection based on simulated GP datasets. Section~\ref{sec:app_study} provides a comparison with methods commonly used in machine learning for variable selection and prediction based on real datasets, including an example with $n = 10^6$ and $d = 10^3$. Section~\ref{sec:conclusion} concludes the paper. The code for replicating the numerical results in this paper are published at \url{https://github.com/katzfuss-group/Vecchia_GPR_var_select}.


\section{Review \label{sec:review}}

\subsection{GP regression and ARD kernels}

We consider the standard GP regression model \citep[e.g.,][]{Rasmussen2006}:
\[
y_i = f(\bx_i) + \epsilon_i, \qquad i=1,\ldots,n,
\]
where $y_i$ is the $i$-th response observed at the $d$-dimensional covariate vector $\bx_i \in \domain \subset \mathbb{R}^d$, $f(\cdot)\sim \GP(0, K)$ is a GP with zero mean and a positive-definite covariance or kernel function $K: \domain \times \domain \rightarrow \mathbb{R}$, and $\{\epsilon_i \sim \normal(0,\tau^2)\}$ are independent noise terms. Then, the vector of responses, $\by = \big(y_1, \ldots, y_n\big)^\top$, at input values $\bx_1,\ldots,\bx_n$ follows an $n$-variate Gaussian distribution, $\normal_n(\bfzero,\bfSigma)$, with covariance matrix
$
\bfSigma = \big(K(\bx_i,\bx_j)\big)_{i,j=1,\ldots,n} + \tau^2 \bI_n,
$ 
whose $(i,j)$-th entry describes the covariance between responses $y_i$ and $y_j$ as a function of their corresponding covariate vectors $\bx_i$ and $\bx_j$. Throughout, we assume a centered response vector $\by$ and a zero mean structure; if desired, a (non-zero) linear mean structure can be profiled out during maximum likelihood estimation \citep{Guinness2019}. 

An automatic relevance determination (ARD) kernel \citep{neal1996bayesian} is an anisotropic kernel that assigns each covariate a separate parameter, controlling its impact in the covariance structure. Specifically, we assign a separate relevance (i.e., inverse range) parameter $r_l \geq 0$ to each input dimension $l$:
\begin{equation}
    \label{equ:scaleddist}
    \textstyle
    K(\bx_i,\bx_j) = \tilde K(q^{\br}(\bx_i,\bx_j)), \qquad
    q^{\br}(\bx_i,\bx_j)^2 = \sum_{l=1}^d r_l^2 (x_{i,l}-x_{j,l})^2,
\end{equation}
where the superscript $\br$ emphasizes the dependence of the distance $q$ on the relevances $\br=(r_1,\ldots,r_d)^\top$. Note that $r_l = 0$ is equivalent to deselecting the $l$-th covariate. In \eqref{equ:scaleddist}, $\tilde K$ can be any isotropic kernel that is valid in $\mathbb{R}^d$; for our numerical results, we used a Mat\'ern covariance kernel with smoothness $2.5$ as recommended in Chapter~4 of \cite{Rasmussen2006}:
\begin{equation}
    \label{equ:matern25}
    \tilde K(q) = \sigma^2(1 + q + q^2 / 3) \exp\left( - q \right),
\end{equation}
where $\sigma^2$ is the variance parameter. Our model depends on unknown parameters $\bftheta = (\sigma^2, \br^2, \tau^2)$, whose inference is usually achieved by maximum likelihood estimation (MLE). We denote by $\br^2$ the element-wise square of $\br$; we use the squared relevance (SR) as the optimization parameters for the purpose of variable selection, which is explained in Section~\ref{sec:GPR_var_select}. Computing the exact GP density, $\dens_{\bftheta}(\by) = \normal_n(\by | \mathbf{0}, \bfSigma_{\bftheta})$, requires $\order(n^3)$ time and $\order(n^2)$ memory, often becoming infeasible for $n > 10{,}000$.

\subsection{Review of scaled Vecchia}
\label{subsec:review_Vecchia}

We use the (scaled) Vecchia approximation to tackle GP regressions with large $n$ (e.g., $n > 10^4$), because it can achieve higher approximation accuracy while having the same linear complexity compared with other state-of-the-art GP approximations. The original Vecchia approximation \citep{Vecchia1988} starts from the conditional representation of the density function, $\dens_{\bftheta}(\by) = \prod_{i=1}^n \dens_{\bftheta}(y_i | \by_{1 : (i - 1)})$, and truncates the conditioning sets to sets $c(i)$ with a maximum of $m \ll n$ elements:
\begin{equation}
    \label{equ:vecchia}
    \textstyle  \hat{\dens}_{\bftheta}(\by) = \prod_{i=1}^n \dens_{\bftheta}(y_i | \by_{c(i)}) = \normal_n(\bfzero,\hat\bfSigma).
\end{equation}
The Vecchia approximation has several attractive properties. It partitions the $n$-dimensional GP density into $n$ computationally independent univariate conditional densities, and hence results in $n$ parallel computations each requiring only $\order(m^3)$ time, where even small $m \ll n$ can achieve high accuracy due to the screening effect \citep{Stein2011}. As indicated by \eqref{equ:vecchia}, the approximation also implies a joint Gaussian distribution, whose inverse Cholesky factor $\hat\bfSigma^{-1/2}$ is sparse with fewer than $nm$ nonzero entries \citep[e.g.,][]{Katzfuss2017a}. Furthermore, Vecchia approximation produces the smallest KL divergence from $\dens_{\bftheta}(\by)$ subject to certain sparsity constraints on $\hat\bfSigma^{-1/2}$ \citep{Schafer2020} and can achieve $\epsilon$-accurate approximations with $m = \order(\log^d(n))$ for certain Mat\'ern-type kernels up to edge effects \citep{Schafer2020}.

The accuracy of Vecchia approximations depends on the ordering of $\by$ and the choice of $\{c(i)\}$; the scaled Vecchia approximation in \citet{Katzfuss2020} takes varying relevances of the covariates into account. Specifically, the scaled Vecchia approximation uses the maximum-minimum distance ordering (MM) and the nearest-neighbor conditioning (NN) based on the scaled distances $q^{\br}(\bx_i, \bx_j)$ between $y_i$ and $y_j$. MM is a sequential ordering that selects each response to maximize the minimum distance toward previous responses in the ordering, and NN chooses the $\min(i - 1, m)$ nearest responses of $y_i$ among $\{y_1, \ldots, y_{i - 1}\}$ as $\by_{c(i)}$. MM and NN can be obtained in quasilinear time in $n$ \citep{Schafer2017, Schafer2020}. We use $\hat{\dens}_{\bftheta}^{\tilde{\br}}(\by)$ to represent the scaled Vecchia likelihood evaluated at $\bftheta$ with MM and NN computed based on $q^{\tilde\br}$, where $\tilde\br$ does not necessarily have to take on the same values as the $\br$ indicated by $\bftheta$. 

Another attractive property of the Vecchia approximation is that many existing GP approximations, including FIC and PIC, can be viewed as its special cases corresponding to particular choices of the ordering and conditioning \citep{Katzfuss2017a}; however, the scaled MM and NN choices in scaled Vecchia can be much more accurate. To demonstrate this, we used a numerical experiment to compare FIC, FITC (with optimized pseudo-inputs), PIC, Vecchia (with MM and NN based on $q^{\mathbf{1}}$), and scaled Vecchia approximations in terms of their KL divergence from an exact multivariate Gaussian distribution
(see details in Appendix \ref{app:pic}). Figure~\ref{fig:Vecchia_FIC_PIC_compare} shows the results for the comparison with $n = 5{,}000$, $d = 10$, $\sigma^2 = 1$, $\br = (10, 5, 2, 1, 0.5, 0, \ldots, 0)^\top$, $\tau^2 = 0$, averaged over ten repetitions. 
\begin{figure}
\centering
	\begin{subfigure}{.75\textwidth}
	\centering
 	\includegraphics[width =.98\linewidth]{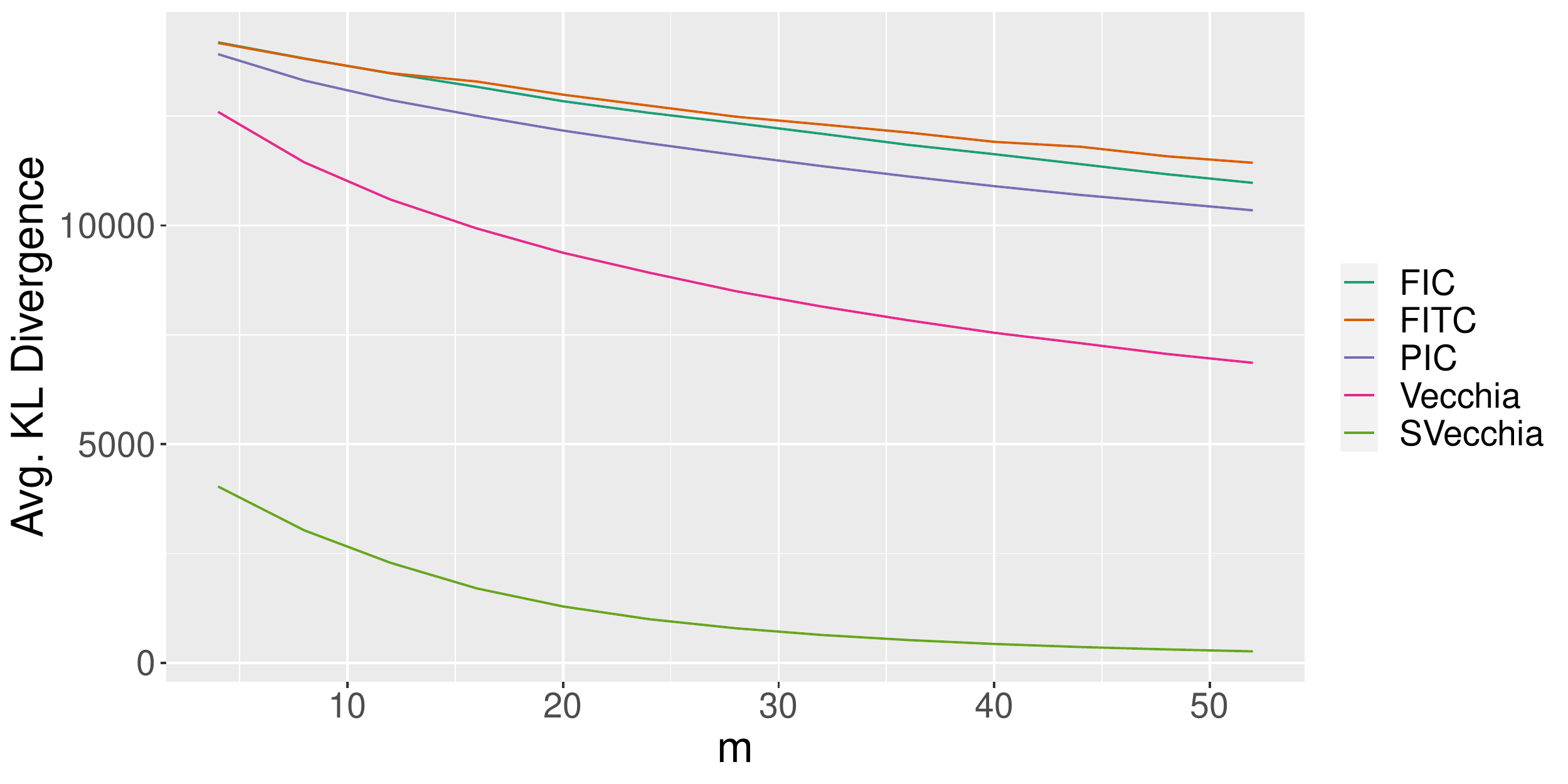}
	\end{subfigure}%
  \caption{Approximation accuracy (in terms of KL divergence from the true GP density) versus conditioning-set size, for five GP approximations, namely FIC, FITC, PIC, Vecchia and Scaled Vecchia (SVecchia)}
  \label{fig:Vecchia_FIC_PIC_compare}
\end{figure}
While Vecchia without scaling outperformed FIC, FITC, and PIC, the scaled Vecchia approach, which will be used in our proposed methods below, resulted in additional improvements of several orders of magnitude.

The construction of the conditioning sets $c(i)$ in the scaled Vecchia approximation can be also applied to posterior prediction to achieve an $\order(m^3)$ complexity at each unknown location. Specifically, the $m$ nearest in-sample neighbors of an unknown location based on $q^{\tilde\br}$ is defined as its conditioning set, based on which the conditional mean and variance is computed. 
Fast computation of the joint posterior predictive distribution at a large set of test inputs is also possible \citep{Katzfuss2018}.

\subsection{Gradient and Fisher information}
\label{subsec:grad_FIM}

The (penalized) negative log-likelihood, $h_{\lambda}^{\tilde\br}(\bftheta) = - \hat{\ell}^{\tilde{\br}}(\bftheta) + w_{\lambda}(\bftheta)$ is typically used as the objective function for parameter inference in GP regression, where here $\hat{\ell}^{\tilde{\br}}(\bftheta) = \log \hat{\dens}_{\bftheta}^{\tilde{\br}}(\by)$ is the log-likelihood under the scaled Vecchia approximation and $w_{\lambda}(\bftheta)$ is a penalty function whose magnitude increases with $\lambda$. Under the Vecchia approximation, not only the log-likelihood but also its first- and second-order information can be computed in parallel and at linear complexity in $n$. Specifically, $\hat{\ell}^{\tilde{\br}}(\bftheta)$ can be decomposed into the sum of $n$ computationally independent terms:
\begin{align}
    \label{equ:obj_SVecchia}
    h_{\lambda}^{\tilde\br}(\bftheta) = - \hat{\ell}^{\tilde{\br}}(\bftheta) + w_{\lambda}(\bftheta) &= -\sum_{i = 1}^{n} \left( \log \dens_{\bftheta}(\by_{\{i\} \cup c(i)}) - \log \dens_{\bftheta}(\by_{c(i)}) \right) + w_{\lambda}(\bftheta).
\end{align}
Based on this expression involving a sum of (log) Gaussian densities, it is straightforward to compute the gradient $\hat{\bg}_{\bftheta}^{\tilde{\br}}$ and the Fisher information matrix (FIM) $-\hat{\bH}_{\bftheta}^{\tilde{\br}}$ of $\hat{\ell}^{\tilde{\br}}(\bftheta)$. Notice that $\hat{\bH}_{\bftheta}^{\tilde{\br}}$ can be used as a surrogate of the Hessian matrix. The computations of $\hat{\bg}_{\bftheta}^{\tilde{\br}}$ and $\hat{\bH}_{\bftheta}^{\tilde{\br}}$ are $\order(nm^3d)$ and $\order(nm^2d^2)$, respectively, based on the closed-form formula for multivariate normal gradient and FIM; refer to \cite{Guinness2019} and the R package `GpGp' \citep{Guinness2016a} for the computation details.

The availability of the second-order information under the Vecchia approximation benefits the convergence rate of parameter inference. Along this direction, a state-of-the-art method is the Fisher scoring algorithm proposed in \citet{Guinness2019} that substitutes the Hessian matrix in the natural gradient descent with FIM to achieve a quadratic convergence rate:
\begin{align}
    \bftheta^{(\iota + 1)} = \bftheta^{(\iota)} - \left(\hat{\bH}_{\bftheta^{(\iota)}}^{\tilde\br^{(\iota)}}\right)^{-1} \hat{\bg}_{\bftheta^{(\iota)}}^{\tilde\br^{(\iota)}},
    \label{equ:Fisher_update}
\end{align}
where the superscript denotes the parameter estimates at the $\iota$-th iteration. However, it is not ideal for constrained optimization. Specifically, Fisher scoring uses variable transformation (e.g., logarithm) to enforce positivity constraints, and so it is typically impossible for optimization parameters to reach boundary values (i.e., zero), which is crucial for variable deselection. We introduce a new second-order optimization algorithm that addresses this limitation in Section~\ref{subsec:QCCD}.

\section{Scalable GP Regression and Variable Selection}
\label{sec:GPR_var_select}

\subsection{Overview of VGPR
\label{subsec:vgpr}}

Algorithm~\ref{alg:VGPR} contains a high-level overview of our VGPR algorithm for scalable variable selection and model estimation in GP regression, with subsequent sections providing details and theoretical and numerical support.
VGPR traverses the regularization path of the penalized log-likelihood from strong to weak penalization until a stopping criterion based on an out-of-sample (OOS) score is reached (Section \ref{subsec:traverse_path}). For a given penalization level, VGPR conducts a forward-backward-selection procedure (Section \ref{subsec:forwardbackward}), which iteratively adds covariates to a candidate set based on the gradient with respect to the squared relevances (Section \ref{subsec:covariate_select_thm}) and deselects covariates through QCCD optimization (Section \ref{subsec:QCCD}). We introduce an iterative adaptive bridge penalty (Section \ref{subsec:adapt_bridge}) and provide further speed-ups via an unbiased mini-batch subsampling method (Section \ref{subsec:minibatch}), resulting in a computational complexity that is essentially independent from $n$ and $d$ (Section \ref{subsec:complexity}).

\begin{algorithm}[h]
    \caption{VGPR}
    \label{alg:VGPR}
    \KwInput{$\hat{\ell}^{\tilde{\br}}(\bftheta), w_{\lambda}(\bftheta), \lambda_0, k$}
    \begin{algorithmic}[1]
        \STATE Initialize $\bftheta$ with $\br$ set to $\mathbf{0}^{+}$ and $\tilde{\br} \gets \br, \lambda \gets \lambda_0$, $\zeta \gets \phi$
        \WHILE{OOS score improves}
            \STATE $(\bftheta, \zeta) \gets$ forward-backward$(\hat{\ell}^{\tilde{\br}}(\bftheta), w_{\lambda}(\bftheta), \bftheta, \zeta, k)$ --- see Alg.~\ref{alg:forward_backward}
            \STATE Reduce $\lambda$
        \ENDWHILE
    \end{algorithmic}
\end{algorithm}

\subsection{Traversing the regularization path}
\label{subsec:traverse_path}

VGPR traverses the regularization path of the penalized log-likelihood, $h_{\lambda}^{\tilde\br}(\bftheta)$, from strong to weak penalization (i.e., large to small $\lambda$) until a stopping criterion based on an out-of-sample (OOS) score is reached. We recommend starting with a penalty strength of $\lambda_0 = n$, which is typically sufficient to imply a completely sparse model without any covariates selected. (Otherwise, we simply increase $\lambda$ exponentially until a fully sparse model is obtained.) The regularization path is constructed over a decreasing series of $\lambda$, for example, a geometric series with a common ratio of $1/2$.
VGPR stops when an out-of-sample (OOS) score such as mean-squared error fails to show obvious improvement. 

Figure~\ref{fig:reg_path} illustrates the regularization path computed by VGPR using $n = 10^4$ responses and $d = 10^3$ covariates under the bridge penalty (see Section~\ref{subsec:adapt_bridge}).
\begin{figure}
\centering
	\begin{subfigure}{.5\textwidth}
	\centering
 	\includegraphics[width =.98\linewidth]{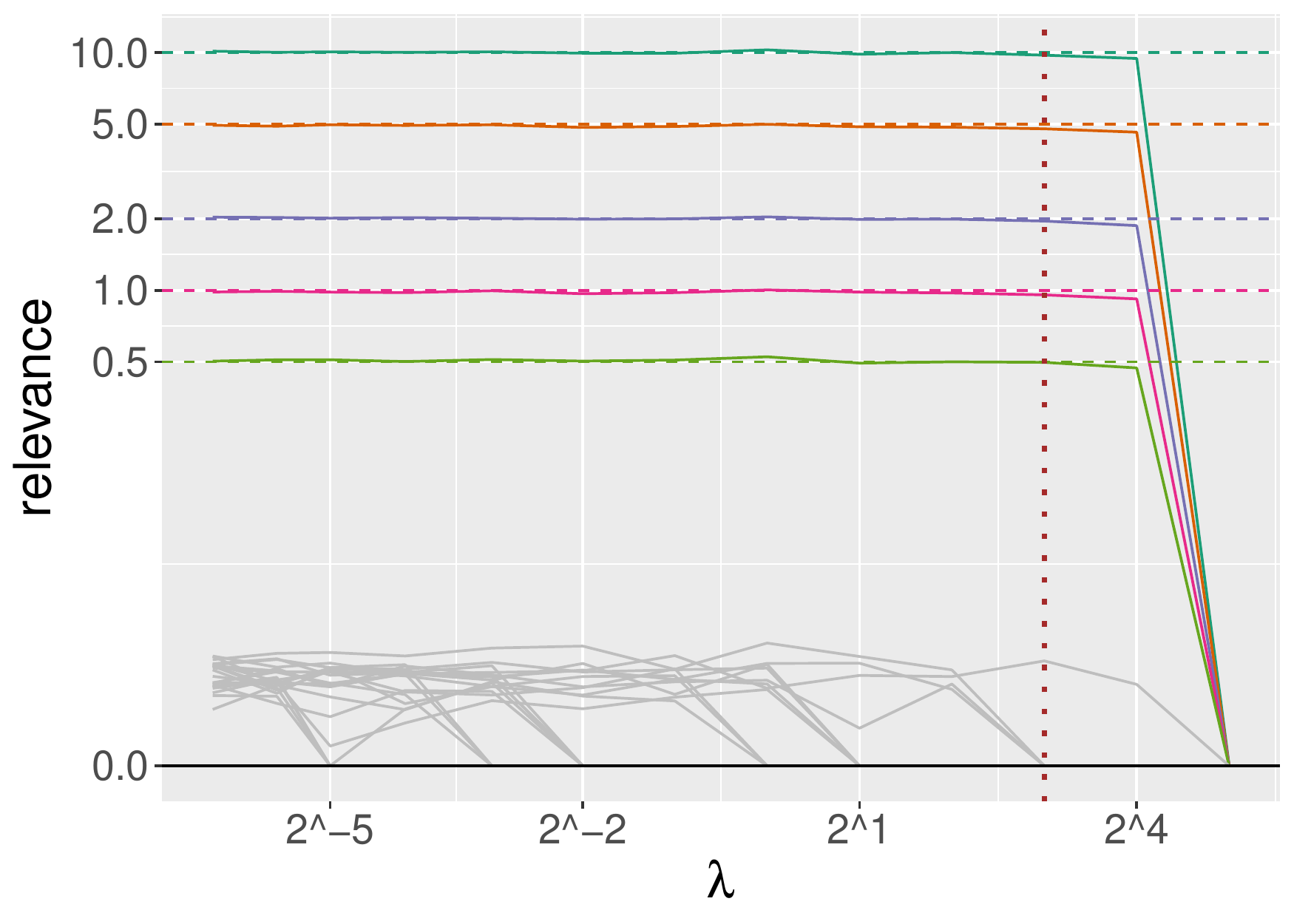}
	\caption{Independent covariates}
	\label{fig:QCCD_bridge_5000_1000}
	\end{subfigure}%
    \hfill
	\begin{subfigure}{.5\textwidth}
	\centering
 	\includegraphics[width =.98\linewidth]{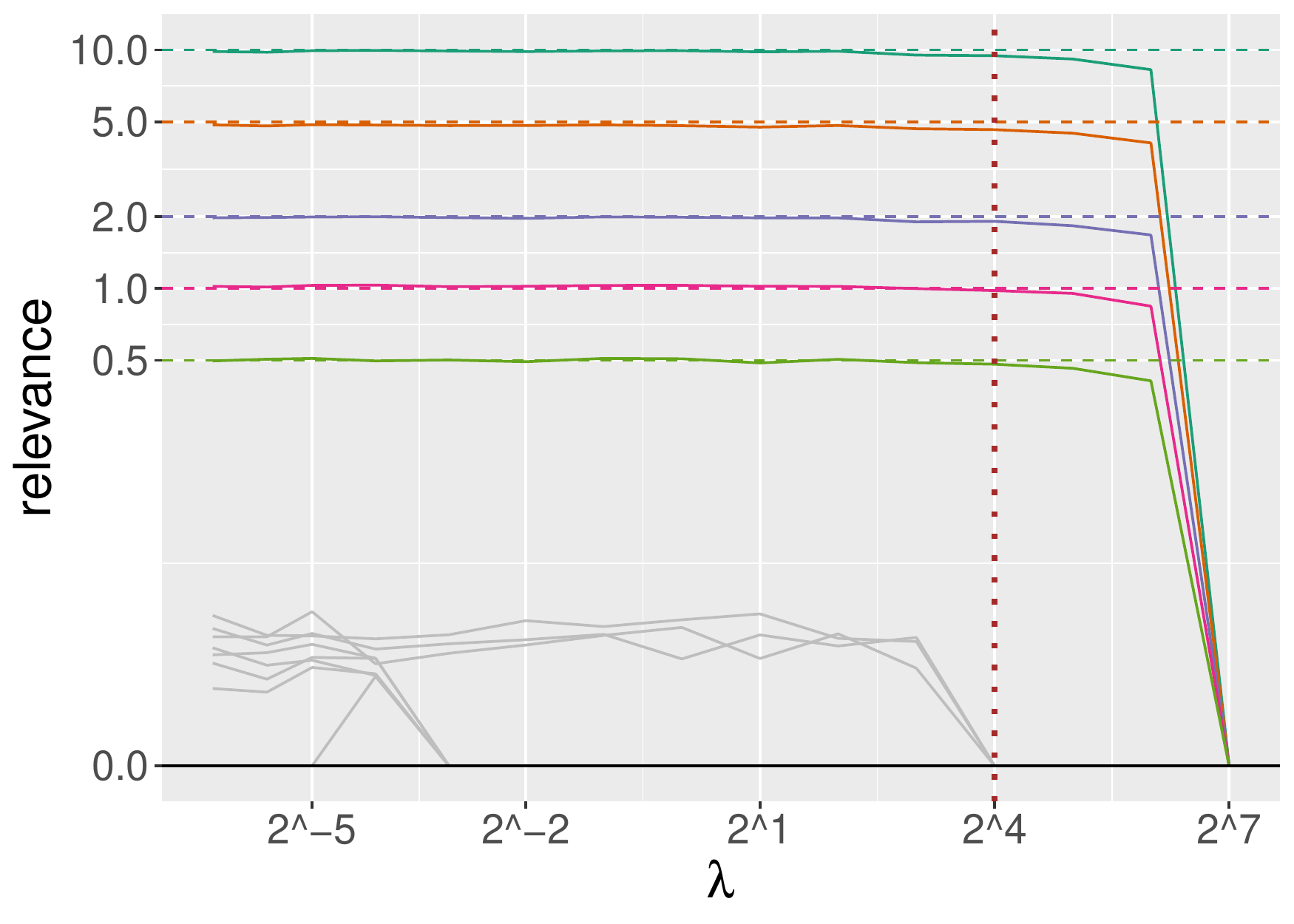}
	\caption{Dependent covariates}
	\label{fig:QCCD_bridge_dep_5000_1000}
	\end{subfigure}%
	
    \caption{Regularization path computed by VGPR using simulated independent or dependent covariates. The relevance parameters of the true covariates are color-coded, and their true values are marked by horizontal colored dashed lines. The fake covariates, whose true relevance parameters are zero, are colored in grey. The vertical red dotted lines mark the optimal model indicated by the stopping condition.}
    \label{fig:reg_path}
\end{figure}
The covariance kernels used for dataset simulation in Sections~\ref{sec:GPR_var_select} and \ref{sec:sim_study} are the Mat\'ern covariance kernel defined by \eqref{equ:matern25} and parameterized by:
\begin{equation}
    \sigma^2 = 1, \quad \tau^2 = 0.05^2, \quad [r_1^2, r_2^2, r_3^2, r_4^2, r_5^2] = [10^2, 5^2, 2^2, 1^2, 0.5^2], \quad r_l^2 = 0 \mbox{ if } l > 5,
    \label{equ:true_val_five}
\end{equation}
unless specified otherwise. The covariates are generated either independently from the Latin hypercube or dependently from a multivariate normal distribution with a constant correlation of $0.9$ and normalized to have a standard deviation of one. Our Vecchia approximation uses a maximum conditioning set size of $m = 100$. A quarter of the responses were set aside to compute the OOS RMSE based on which, the stopping condition was defined as that the OOS RMSE improves less than $1\%$ after any new covariate is selected. The OOS sample size of $n/4$ and the $1\%$ OOS score threshold are used throughout this paper and are generally recommended as default values.

In Figure \ref{fig:reg_path}, the true covariates and their relevance parameters were correctly selected and well estimated, respectively. All fake covariates, except for one when using independent covariates, were filtered out, highlighting the efficacy of VGPR in variable selection even given a large pool of highly correlated covariates. Moreover, the number of optimization parameters was always kept at $\order(d_0)$ until the stopping condition was reached. Also due to the small number of optimization parameters, VGPR completed the model estimation within minutes.

\subsection{Forward-backward selection \label{subsec:forwardbackward}}

To keep the ``active'' set of covariates small when running optimization, VGPR keeps a candidate set of covariates $\zeta \subset \{1,2,\ldots,d\}$ representing the covariates currently selected. Assuming model sparsity, the size of $\zeta$ can be kept much smaller than $d$. Given a current $\zeta$, standard forward selection would fit $\order(d)$ models with covariates $\zeta \cup l$ for each $l \notin \zeta$, but this procedure is prohibitively expensive for large $d$.

Instead, we propose a forward-backward-selection algorithm, provided in Algorithm~\ref{alg:forward_backward}, to find the optimal model under each $\lambda$. The algorithm iteratively performs a forward step and a backward step. The forward step adds to $\zeta$ a small number $k$ of ``promising'' covariates corresponding to the $k$ largest entries in the squared-relevance gradient (SR-gradient), given by the derivatives of $\ell(\bftheta) = \log \dens_{\bftheta}(\by)$ with respect to each $r_l^2$ with $l \notin \zeta$, evaluated at the current estimates of the relevances (i.e., $r_l = 0$ for $l \notin \zeta$). For example, we set $k = 3$ in Figure~\ref{fig:reg_path}. We provide numerical and theoretical support for the forward step in Section \ref{subsec:covariate_select_thm}.
\begin{algorithm}[h]
    \caption{Forward-backward selection}
    \label{alg:forward_backward}
    \KwInput{$\hat{\ell}^{\tilde{\br}}(\bftheta), w_{\lambda}(\bftheta), \bftheta, \zeta, k$}
    \begin{algorithmic}[1]
        \WHILE{OOS score improves}
            \STATE $\underline{\color{blue} \mathcal{S} \gets \mbox{mini-batch subsampling},}$ $\hat{\bg}_{\br^2} \gets \frac{\partial \hat{\ell}^{\tilde{\br}}(\bftheta\ \underline{\color{blue}\mid \mathcal{S}})}{\partial \br^2}$ \label{stp:grad_r}
            \STATE Define $\Delta \zeta$ as the indices of the $k$ largest coefficients in $\hat{\bg}_{\br^2}[-\zeta]$
            \STATE $\zeta \gets \zeta \bigcup \Delta \zeta$, initialize $\br[\Delta \zeta]$, $\tilde{\br} \gets \br$
            \STATE $h^{\tilde{\br}}_{\lambda, \zeta}(\bftheta_{\zeta}) \gets \hat{\ell}^{\tilde{\br}}_{\zeta}(\bftheta_{\zeta}) + \lambda w_{\lambda}(\bftheta_{\zeta})$, $\bftheta_{\zeta} \gets \mbox{QCCD}(h^{\tilde{\br}}_{\lambda, \zeta}, \bftheta_{\zeta}, \bfzero)$  --- see Alg.~\ref{alg:QCCD} \label{stp:fb_par_est}
            \STATE Remove covariates with zero relevance from $\zeta$ \label{line:forward_backward_backstep}
        \ENDWHILE
        \STATE \Return $\bftheta$ and $\zeta$
    \end{algorithmic}
\end{algorithm}
After the forward step, we run a backward step on the new $\zeta$ via our QCCD algorithm (see Section \ref{subsec:QCCD}), which finds the new parameter estimates using a warm start based on the previous estimates and potentially deselects covariates by returning estimates of zero for some SRs.
The forward-backward procedure for a given $\lambda$ value stops (and VGPR moves on to a smaller $\lambda$) based on the same stopping criterion as in Algorithm~\ref{alg:VGPR}, using the OOS score. 

We now provide more notational details on Algorithm \ref{alg:forward_backward}. We use square brackets for indexing, with negative indices corresponding to dropped elements. The parts in blue font (in all algorithms) provide the mini-batching modifications to be discussed in Section~\ref{subsec:minibatch}. MM and NN are implicitly updated at each occurrence of $\tilde{\br} \gets \br$, which improves the accuracy of the scaled Vecchia approximation $\hat{\ell}^{\tilde{\br}}(\bftheta)$. 
In Line~\ref{stp:fb_par_est}, we use the $\zeta$ subscript to indicate the parameter vector, the log-likelihood function, and the objective function defined over the subset of covariates in $\zeta$, as opposed to all covariates, which reduces the number of parameters involved in QCCD.
Notice that $\bftheta_{\zeta}$ is viewed as a subvector of $\bftheta$ and the assignment to the former indicates changes to the latter as well, which implies warm starts and avoids local optima.

\subsection{Numerical and theoretical support for gradient-based covariate selection}
\label{subsec:covariate_select_thm}

The SR-gradient can be used to order the covariates' relevance levels in the ARD model. Specifically, assuming that the SRs of the covariates in $\zeta$ are fixed at their correct values, the derivatives of $\ell(\bftheta) = \log \dens_{\bftheta}(\by)$ with respect to the remaining SRs evaluated at zero can be used to rank the unselected covariates in $\zeta^{\mathrm{C}}$. 

We illustrate this idea using an example of selecting $d_0 = 5$ true covariates from $d = 10^3$ total covariates, shown in Figure~\ref{fig:gradient_magnitude}, with $\sigma^2$ and $\tau^2$ are fixed at their true values. It is evident that true covariates (with $r_{l0}^2 > 0$) have bigger coefficients in the SR-gradient. In fact, the magnitudes of the coefficients reflect the magnitudes of $\{r_{l0}^2\}$. This conclusion is valid even assuming strong dependence among covariates or using the gradient under mini-batch subsampling (see Section~\ref{subsec:minibatch}). The full dataset has $n = 5{,}000$ responses and the mini-batch size is $\check{n}=128$. We used the derivatives under the scaled Vecchia approximation (i.e., $ \hat\ell^{\tilde{r}}(\bftheta)$) to substitute those of $\ell(\bftheta)$, indicating sufficient accuracy from the scaled Vecchia approximation.
\begin{figure}
\centering
	\begin{subfigure}{.5\textwidth}
	\centering
 	\includegraphics[width =.98\linewidth]{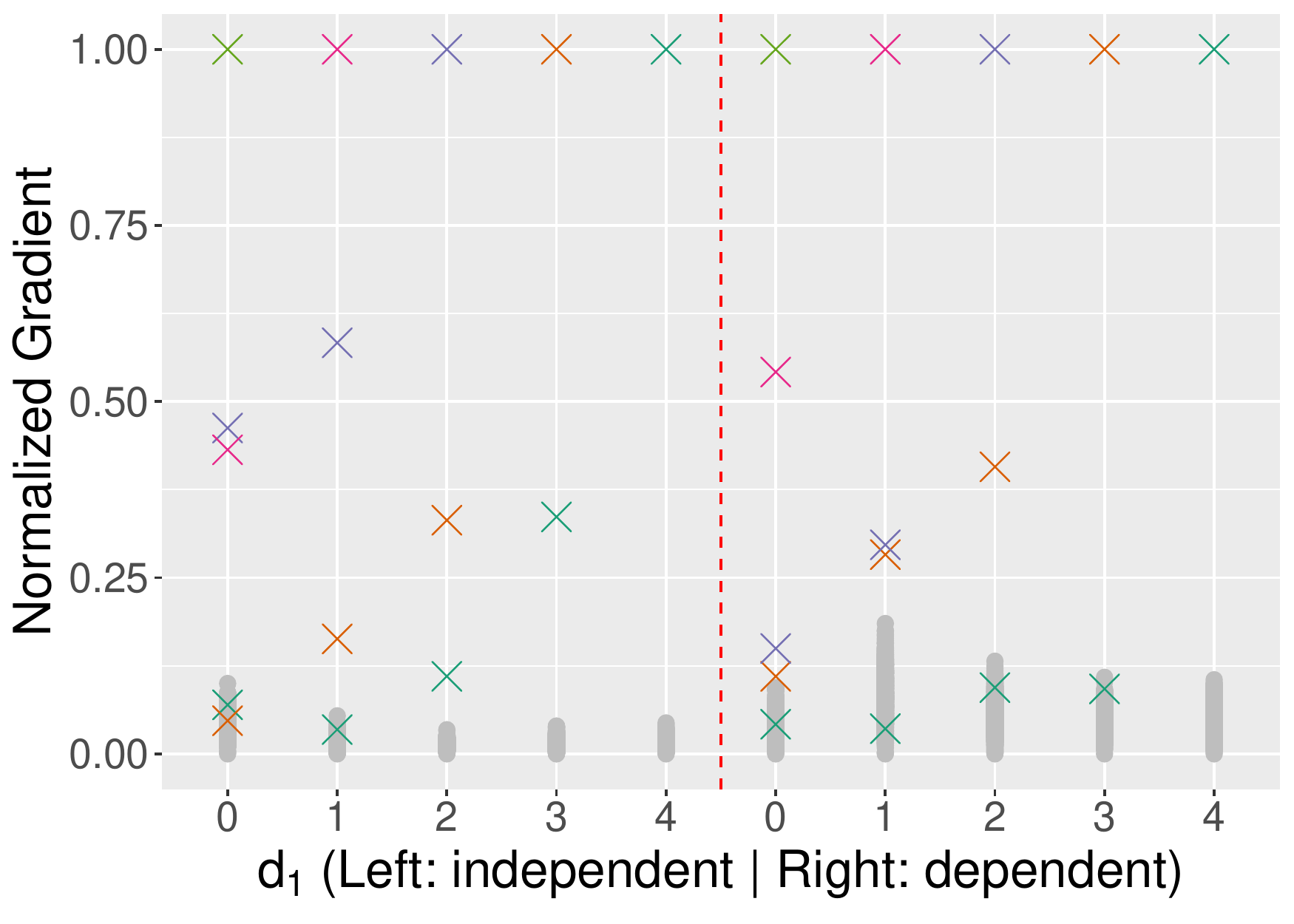}
	\caption{Normalized gradient without mini-batching}
	\label{subfig:gradient_magnitude}
	\end{subfigure}%
    \hfill
	\begin{subfigure}{.5\textwidth}
	\centering
 	\includegraphics[width =.98\linewidth]{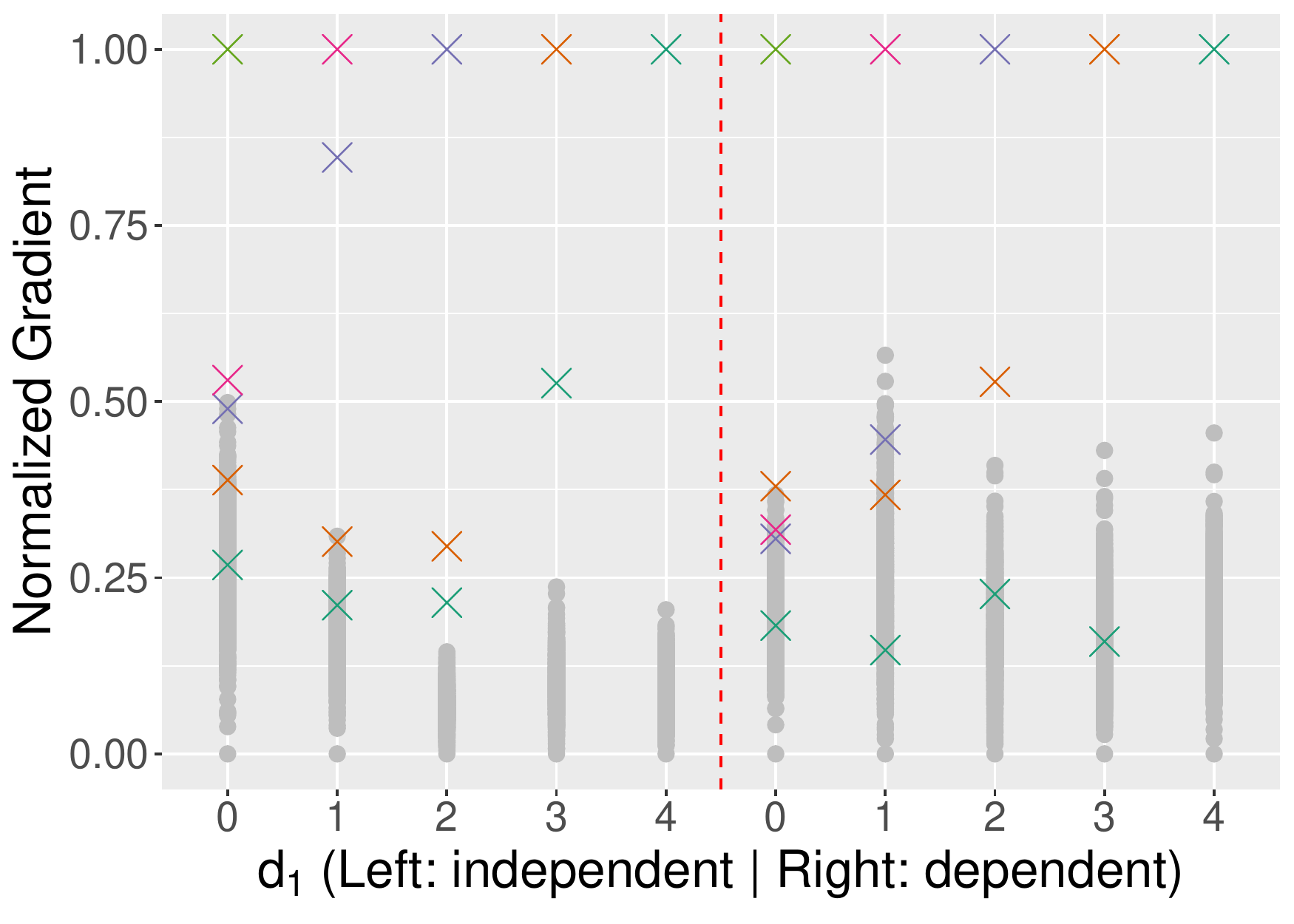}
	\caption{Normalized gradient with mini-batching}
	\label{subfig:gradient_magnitude_minibat}	
	\end{subfigure}%
	
  \caption{Relative magnitudes of the coefficients in the SR-gradient. The number of covariates $d = 10^3$, among which the first five are true (i.e., with positive true SRs). For each $d_1=0,1,2,3,4$, we assume that only the first $d_1$ true covariates are selected and their SRs are correctly estimated. The SRs of unselected covariates are zero. The coefficients in the gradient are normalized to $[0, 1]$. The first five coefficients in the SR-gradient are marked by colored crosses and the rest by grey dots. Only coeffcients corresponding to unselected covariates are plotted to align with goal of variable selection. The red dashed line separates scenarios with independent and dependent covariates. Notice that some colored crosses are covered by grey dots and that the coefficients for unselected true covariates were typically bigger than those for fake covariate.}
  \label{fig:gradient_magnitude}
\end{figure}

In the remainder of this section, we provide theoretical support for why the SR-gradient can be used for variable selection. The following notations are used in the theoretical results and their derivations:
\begin{align*}
    \begin{array}{ll}
        \br_0 & \mbox{the true relevance vector $[r_{10}, r_{20}, \ldots, r_{d0}]$} \\
        d_0 &  \mbox{the number of true covariates (i.e., $r_{l0} > 0$ if $l \le d_0$ $r_{l0} = 0$ otherwise)}\\
        d_1 & \mbox{an integer between $0$ and $d_0$, $0 < d_1 < d_0 < d$} \\
        \br_1 & [r_{10}, \ldots, r_{d_10}, 0, \ldots, 0]^\top \\
        (\sigma_0, \tau_0), (\sigma_1, \tau_1) & \mbox{the true and an arbitrary values for $(\sigma, \tau)$} \\
        \bfSigma, \bfSigma_0, \bfSigma_1 & \mbox{covariance matrix and its values evaluated at $(\sigma_0, \br_0, \tau_0)$ and $(\sigma_1, \br_1, \tau_1)$} \\
        \tilde\bfSigma, \tilde\bfSigma_0, \tilde\bfSigma_1 & \mbox{correlation matrix and its values evaluated at $(1, \br_0, 0)$ and $(1, \br_1, 0)$}
    \end{array}
\end{align*}
In this section, the expectations are taken with respect to both $\{\bx_i\}_{i = 1}^n$ and $\{y_i\}_{i = 1}^n$. 

\begin{proposition}
    \label{thm:sqexp}
    Assume that $K(\bx_i, \bx_j) = \exp(-q^{\br}(\bx_i, \bx_j)^2)$ and that $\{x_{il}\}_{i = 1, \ldots, n, l = 1, \ldots, d}$ have i.i.d.\ normal or uniform distributions. When evaluated at $(\sigma_1, \br_1, \tau_1)$, $E[\frac{\partial \ell}{\partial r_{l_1}^2}] > E[\frac{\partial \ell}{\partial r_{l_2}^2}]$.
\end{proposition}
Proposition~\ref{thm:sqexp} suggests that under the squared exponential kernel, when an arbitrary number of SRs are at their true values while the others at zero, the order of the SR-gradient coefficients indicates the relevance order of the covariates. While the condition on $\br_1$ in Proposition~\ref{thm:sqexp} is somewhat restrictive, we conjecture that when the gradient is evaluated at $\br$ no greater than $\br_0$ coefficient-wise, the above conclusion still holds, which can be readily shown if we assume $E[\frac{\partial \ell}{\partial r_{l_1}^2}] - E[\frac{\partial \ell}{\partial r_{l_2}^2}]$, evaluated at $(\sigma_1, \br, \tau_1)$, changes monotonically with each coefficient in $\br$. In general, numerical examples in Sections~3.1, 3.7, and 4 suggest that the conditions in Proposition~\ref{thm:sqexp} can be relaxed and the result still holds.

Based on Proposition~\ref{thm:sqexp}, two corollaries addressing the initialization of $\br$ and correlated fake covariates, respectively, can be derived. 
\begin{corollary}
    \label{cor:sqexp_at_zero}
    Assume that $K(\bx_i, \bx_j) = \exp(-q^{\br}(\bx_i, \bx_j)^2)$ and that $\{x_{il}\}_{i = 1, \ldots, n, l = 1, \ldots, d}$ have i.i.d.\ normal or uniform distributions. When evaluated at $\br \rightarrow \bfzero$:
    \begin{align*}
        \textstyle E[\frac{\partial \ell}{\partial r_{l_1}^2}] \ge E[\frac{\partial \ell}{\partial r_{l_2}^2}]\mbox{ for }0 < l_1 \le d_0\mbox{ and }d_0 < l_2 \le d.
    \end{align*}
\end{corollary}
\begin{corollary}
    \label{cor:sqexp_corr}
    Assume that $K(\bx_i, \bx_j) = \exp(-q^{\br}(\bx_i, \bx_j)^2)$ and that $\{x_{il}\}_{i = 1, \ldots, n, l = 1, \ldots, d}$ have i.i.d.\ normal distributions. Let $\bx_{d + 1}$ be a new covariate constructed as $\rho_1 \bx_{l_1} + \rho_2 \bx_{l_2}$ with $d_1 < l_1 \le d_0 < l_2 \le d$, $\rho_2 > 0$, and $\rho_1^2 + \rho_2^2 = 1$. When evaluated at $(\sigma_1, \br_1, \tau_1)$:
    \begin{align*}
        \textstyle E[\frac{\partial \ell}{\partial r_{l_1}^2}] > E[\frac{\partial \ell}{\partial r_{d + 1}^2}].
    \end{align*}
\end{corollary}

Noticing that $\zeta = \phi$ can be closely approximated by $\br \rightarrow \mathbf{0}$, Corollary~\ref{cor:sqexp_at_zero} indicates that SRs should be initialized to small magnitudes but big enough to avoid numerical singularity (e.g., $10^{-8}$), which is denoted by $\mathbf{0}^{+}$ in Algorithm~\ref{alg:VGPR}. Corollary~\ref{cor:sqexp_corr} suggests that the order of the SR-gradient coefficients can distinguish fake covariates that are correlated with true covariates. Theoretical support for the previous proposition becomes more challenging under general covariance kernels, due to the lack of the separability property and the straight-forward derivative formula. Proposition~\ref{thm:all_kernel} aims to reach the same conclusion for general ARD kernels but uses a first-order approximation of $\bfSigma_0$.
\begin{proposition}
    \label{thm:all_kernel}
    Assume that $\{x_{il}\}_{i = 1, \ldots, n, l = 1, \ldots, d}$ have i.i.d.\ distributions and that $\{r_{l0}\}_{l = d_1 + 1}^{d_0}$ are small enough s.t.\ $\tilde\bfSigma_0$ can be closely approximated by the first-order Taylor expansion of $\tilde\bfSigma$ at $\br_1$:
    \begin{align*}
    \textstyle     \tilde\bfSigma_0 \approx \tilde\bfSigma_1 + \sum_{l = d_1 + 1}^{d_0} \left. \frac{\partial \tilde\bfSigma}{\partial r_l^2} \right|_{\br = \br_1} r_{l0}^2.
    \end{align*}
    When evaluated at $(\sigma_1, \br_1, \tau_1)$:
    \begin{align*}
     \textstyle    E[\frac{\partial \ell}{\partial r_{l_1}^2}] > E[\frac{\partial \ell}{\partial r_{l_2}^2}]\mbox{ for }d_1 < l_1 \le d_0\mbox{ and }d_0 < l_2 \le d.
    \end{align*}
\end{proposition}
The first-order approximation typically holds when $r_{l0} \rightarrow 0, l = d_{1} + 1, \ldots, d_0$ while on the other hand, the expectation of the derivative of $r_l^2$ evaluated at $r_l = 0$ is intuitively positively correlated with $r_{l0}$ (i.e., $\left. E[\frac{\partial \ell}{\partial r_{l}^2}] \right |_{r_{l} = 0} \nearrow r_{l0}$). The two aspects collectively support that the order of the gradient coefficients is indicative for the order of relevance levels of the covariates under general ARD covariance kernels.

\subsection{Quadratic constrained coordinate descent}
\label{subsec:QCCD}

We introduce our quadratic constrained coordinate descent (QCCD) algorithm in the context of minimizing a general objective function $h(\bftheta)$, whose gradient and (positive-definite) negative FIM, denoted by $\bg$ and $\bH$, respectively, can be computed. QCCD is described in Algorithm~\ref{alg:QCCD} with the assumption that parameter constraints are given by their lower bounds $\bb$, but broader constraints on $\bftheta$ can be similarly accommodated.
\begin{algorithm}[h]
    \caption{Quadratic constrained coordinate descent (QCCD)}
    \label{alg:QCCD}
    \KwInput{$h(\cdot), \bftheta_{0}, \bb$}
    \begin{algorithmic}[1]
        \STATE $\bftheta \gets \bftheta_{0}$, $\underline{\color{blue}\alpha \gets 1}$ 
        \WHILE{$\mbox{Not Converged}$}
            \STATE $\underline{\color{blue}\mathcal{S} \gets \mbox{mini-batch subsampling}}$, $\bg \gets \nabla h(\bftheta\ \underline{\color{blue}\mid \mathcal{S}})$, $\bH \gets E[\nabla^2 h(\bftheta\ \underline{\color{blue}\mid \mathcal{S}})]$ \label{stp:QCCD_update1}
            \STATE  $\bftheta_{\mbox{\scriptsize CCD}} \gets \mbox{CCD}(\bftheta, \underline{\color{blue} \alpha}\bg, \bH, \bb)$  --- see Alg.~\ref{alg:CCD} \label{line:CCD_return}
            \STATE $\beta \gets \mbox{argmax}_{\beta \in (0, 1]} \mbox{Armijo}(\beta) > c$, $\bftheta_{\mbox{\scriptsize NEW}} \gets \bftheta + \beta(\bftheta_{\mbox{\scriptsize CCD}} - \bftheta)$
            \color{blue}
            \STATE \underline{\textbf{if} stationarity is detected \label{stp:step_temper_bgn}\textbf{then}}
                \STATE \hspace{\algorithmicindent} \underline{$\alpha \gets \alpha / 2$} 
                \STATE \underline{\textbf{end if}} \label{stp:step_temper_end}
            \color{black}
            \STATE $\bftheta \gets \bftheta_{\mbox{\scriptsize NEW}}$ \label{stp:QCCD_update2}
        \ENDWHILE
        \STATE \Return $\bftheta$
    \end{algorithmic}
\end{algorithm}
$\bftheta_0$ denotes the initial parameter values. Intuitively, QCCD iterates between building a quadratic approximation at the current $\bftheta$,
\begin{equation}
    \hat{h}(\bftheta_{\mbox{\scriptsize NEW}}) = h(\bftheta) + \bg^\top (\bftheta_{\mbox{\scriptsize NEW}} - \bftheta) + \frac{1}{2} (\bftheta_{\mbox{\scriptsize NEW}} - \bftheta)^\top \bH (\bftheta_{\mbox{\scriptsize NEW}} - \bftheta),
    \label{equ:quad_approx}
\end{equation}
and finding the minimum of $\hat{h}(\bftheta_{\mbox{\scriptsize NEW}})$ subject to the constraints on $\bftheta$ using constrained coordinate descent (CCD), described in Algorithm~\ref{alg:CCD}.
\begin{algorithm}[h]
    \caption{Constrained coordinate descent (CCD)}
    \label{alg:CCD}
    \KwInput{$\bftheta, \bg, \bH, \bb$}
    \begin{algorithmic}[1]
        \STATE $\bd \gets \bg - \bH \bftheta$
        \WHILE{$\mbox{Not Converged}$}
            \FOR{$i$ in $1 : \text{length}(\bftheta)$}
                \STATE $\bftheta[i] \gets \max\big((-\bd[i] - \bH[i, -i] \cdot \bftheta[-i]) / \bH[i, i], \bb[i]\big)$
            \ENDFOR
        \ENDWHILE
        \STATE \Return $\bftheta$
    \end{algorithmic}
\end{algorithm}

The CCD algorithm cyclically considers each parameter of $\bftheta$ in a constrained univariate quadratic optimization, where the minimum is analytically available and can be equal to the boundary value. The minimum returned by CCD is subsequently used in a line search subject to the Armijo condition that compares the ratio: 
$$
    \frac{\beta(\bftheta - \bftheta_{\mbox{\scriptsize CCD}}) \cdot \bg}{ h(\bftheta) - h((1 - \beta)\bftheta + \beta \bftheta_{\mbox{\scriptsize CCD}})},
$$ 
with a threshold $c$ to achieve `sufficient decrease' of the objective function (i.e., to avoid unreasonably large steps) and $\bftheta_{\mbox{\scriptsize NEW}}$ is guaranteed to exist subject to mild regularity conditions \citep{kressner2015advanced}.
QCCD is similar to the cyclical coordinate descent algorithm \citep{friedman2010regularization} in terms of building a quadratic approximation and using coordinate descent but has two improvements, namely, the Armijo line search condition and the incorporation of parameter constraints. 

QCCD has the same theoretical convergence rate as Fisher scoring because both find the minimum of the same quadratic approximation, but the former's ability to reach boundary values makes covariate deselection (i.e., $r_l$ being optimized to zero) more straight-forward. 
\begin{figure}
\centering
	\begin{subfigure}{.5\textwidth}
	\centering
 	\includegraphics[width =.98\linewidth]{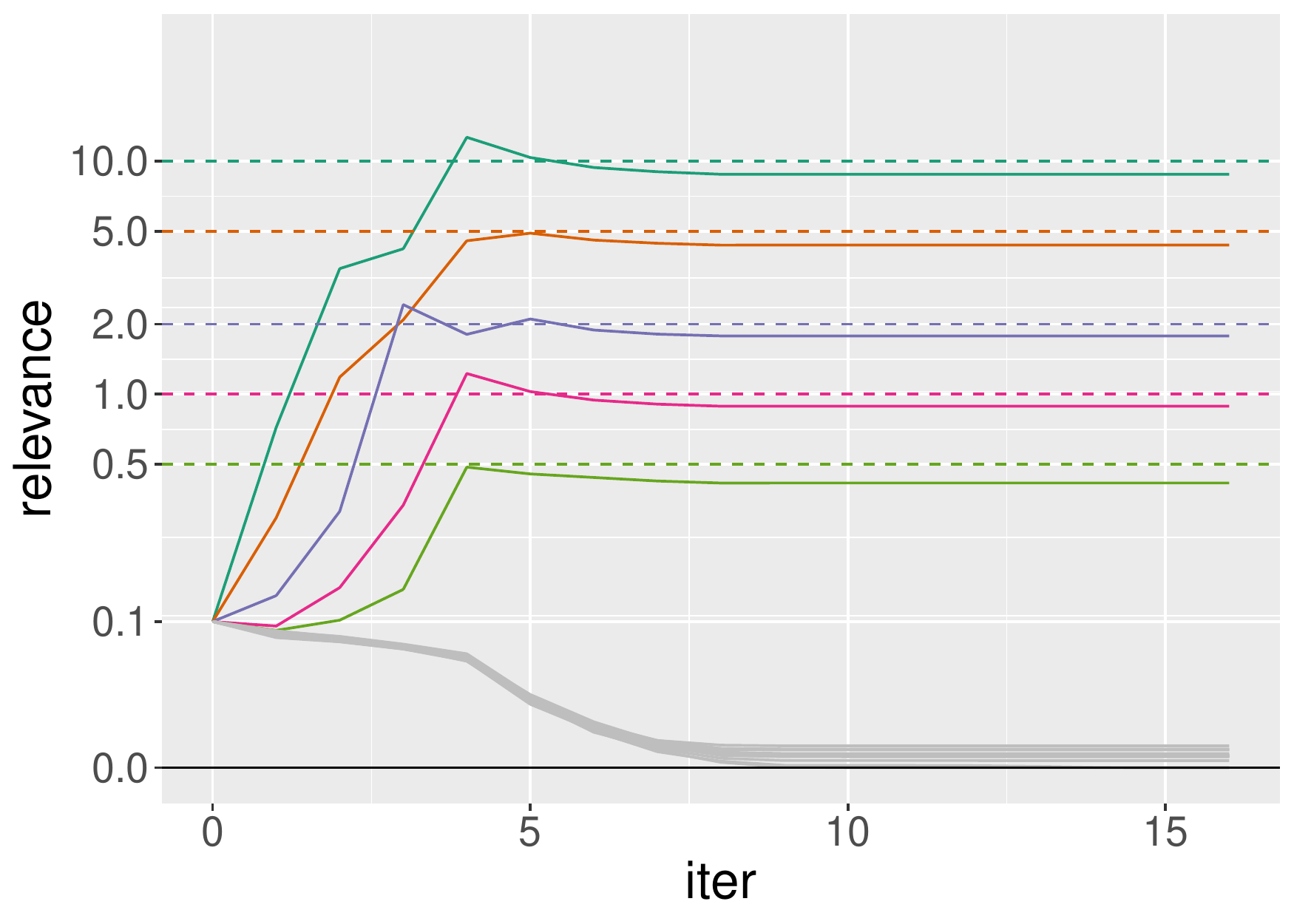}
	\caption{Fisher scoring with $d = 20$ covariates}
	\label{fig:Fisher_scoring_d_20}
	\end{subfigure}%
    \hfill
	\begin{subfigure}{.5\textwidth}
	\centering
 	\includegraphics[width =.98\linewidth]{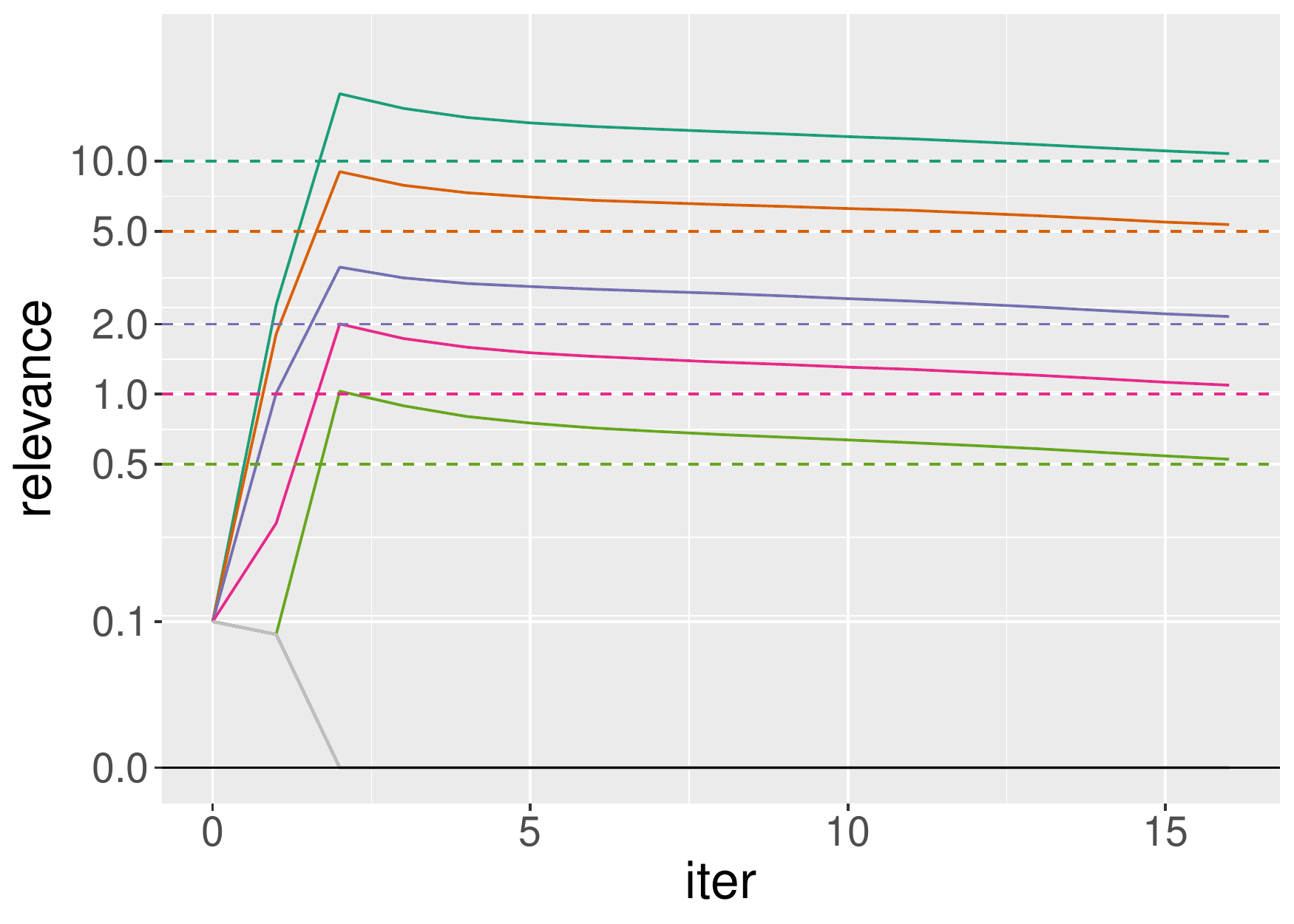}
	\caption{QCCD with $d = 20$ covariates}
	\label{fig:QCCD_small_d}	
	\end{subfigure}%
	
	\smallskip
	
	\begin{subfigure}{.5\textwidth}
	\centering
 	\includegraphics[width =.98\linewidth]{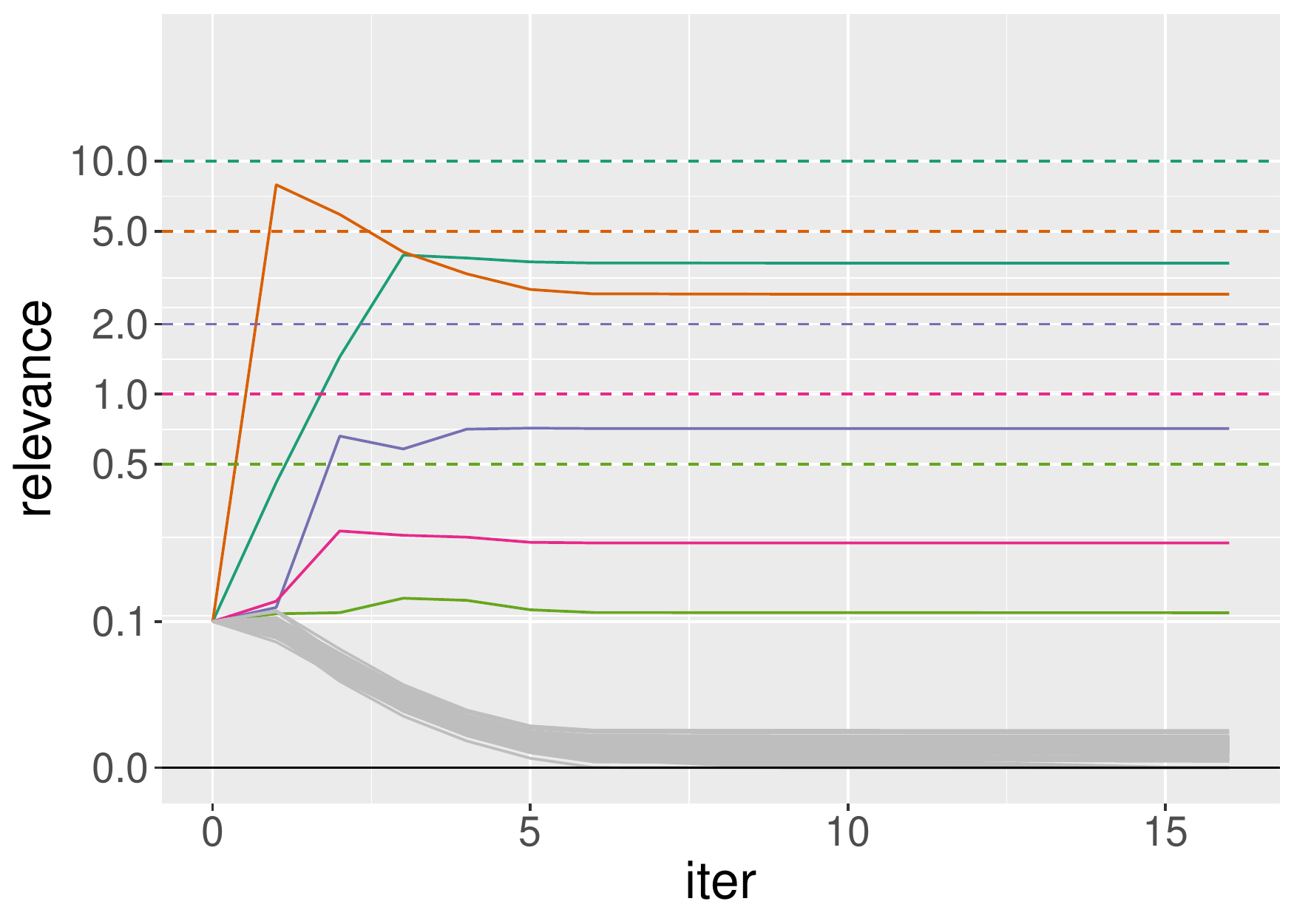}
	\caption{Fisher scoring with $d = 100$ covariates}
	\label{fig:Fisher_scoring_d_100}
	\end{subfigure}%
    \hfill
	\begin{subfigure}{.5\textwidth}
	\centering
 	\includegraphics[width =.98\linewidth]{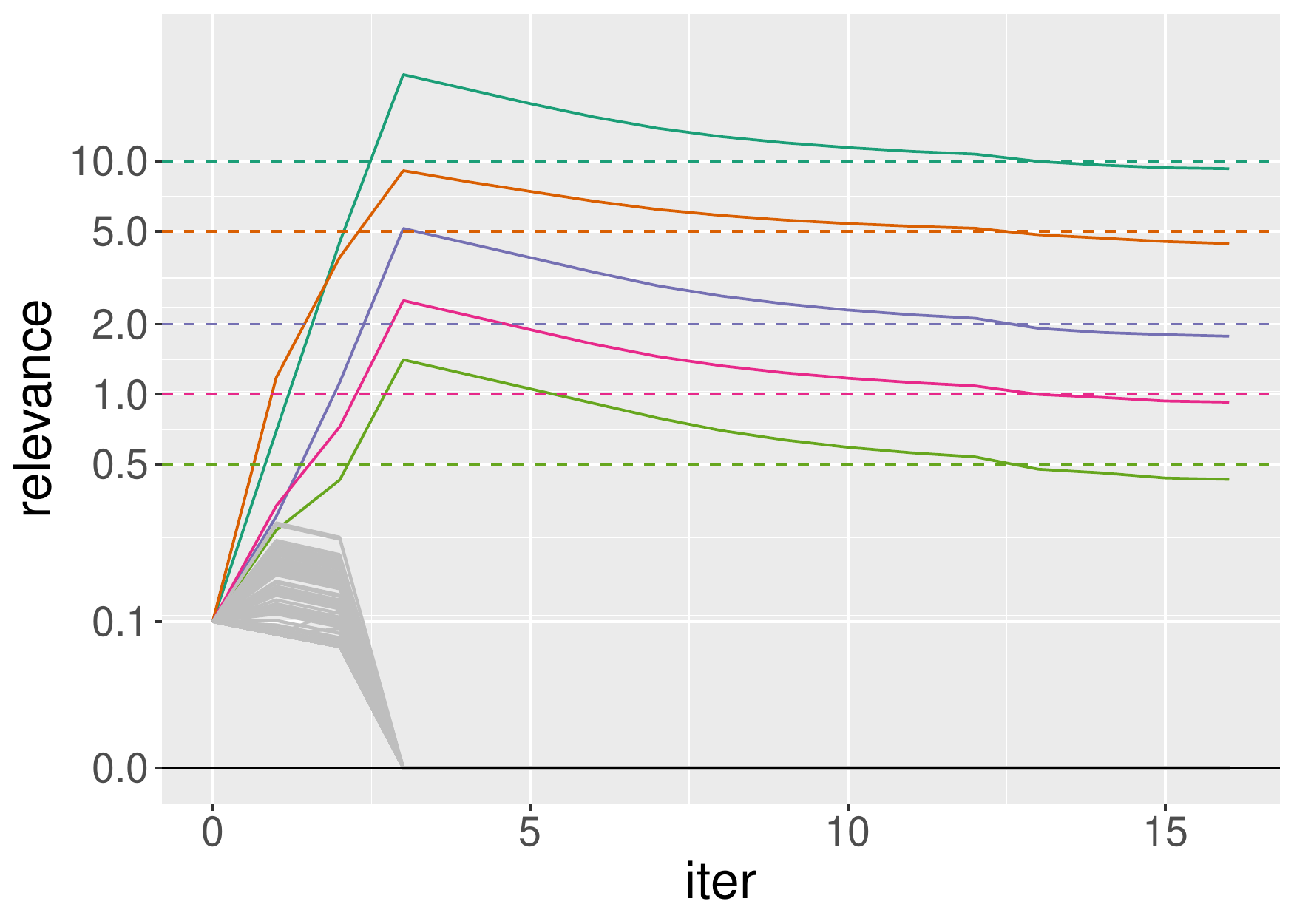}
	\caption{QCCD with $d = 100$ covariates}
	\label{fig:QCCD_d_100}	
	\end{subfigure}
  \caption{Convergence of Fisher scoring and QCCD algorithms. $y$-axis is the relevance on the pseudo-log scale. The true covariates, with relevance $r_l > 0$, are color-coded with their true values marked by the colored dashed lines. The fake covariates, with relevance $r_l = 0$, are colored in grey.}
  \label{fig:cmp_Fisher_QCCD}
\end{figure}
Figure~\ref{fig:cmp_Fisher_QCCD} compares the performance of Fisher scoring and QCCD when $d$ is relatively small (i.e., $d \le 100$). Covariates were independently generated at $n = 10^4$ locations and a bridge penalty with $\lambda = 32$ was used in the objective function, which will be further discussed in Section~\ref{subsec:adapt_bridge}. The relevance vector $\tilde \br$ used for MM and NN was updated together with the updates of $\bftheta$. The parameter estimates from QCCD were closer to the truth than those from Fisher scoring. Further, QCCD was able to deselect all fake covariates, achieving $r_l = 0$ for all $l > 5$, while Fisher scoring was unable to deselect any covariate without setting a truncation level. The ability to automatically deselect covariates becomes increasingly important when addressing GP regressions with larger numbers of covariates.


\subsection{Bridge penalty and its extension}
\label{subsec:adapt_bridge}

The desired properties of the penalty function for GP regression can be different from those for linear regression. \cite{yi2011penalized} compared several penalties in GP regression that include Lasso, SCAD, and bridge penalties, concluding that the bridge penalty has overall the best performance. This agrees with our analysis that unlike linear regression, GP regression automatically avoids improperly large magnitudes of $\br$. Therefore, a penalty function that becomes flat more quickly as the parameter magnitude increases is more suitable for GP regression, leading to higher model sparsity and smaller estimation bias.

However, one issue with the bridge penalty is that its derivative is infinite at zero, and so it is impossible to escape this local optimum for any parameter that reaches zero during optimization. This is especially problematic for the mini-batching procedure to be introduced later, where zero can be reached erroneously due to a ``bad'' batch. Hence, we adopt an iterative adaptive bridge penalty that amounts to a combination of the classic bridge penalty and the iterative adaptive technique in \cite{ziel2016iteratively} and \cite{sun2010genomewide}:
\begin{align}
    w_{\lambda}(\bftheta) = \lambda \sum_{l = 1}^d(c_{\iota, l}^{\kappa} + r_l^2)^{\gamma},\label{equ:adbridge}
\end{align}
where $\iota$ is the iteration number during optimization and $c_{\iota, l}^{\kappa}$ is the sum of the parameter $r_l^2$ over the previous $\kappa$ iterations. In addition to allowing parameters to escape zero values, this adaptive bridge penalty also has the advantage that bigger $r_l$ tends to have bigger $c_{\iota, l}^{\kappa}$, hence weaker penalty and smaller bias. Notice that $\kappa = 0$ corresponds to the classic bridge penalty used in Sections~\ref{subsec:vgpr} and \ref{subsec:QCCD} and that when computing $\hat{\bH}_{\bftheta}^{\tilde{\br}}$, we ignore the second-order information of the penalty function to guarantee the non-negative definiteness of the FIM, which is equivalent to applying a linear approximation to $w_{\lambda}(x)$.

In this paper, we fix $\gamma$ at $0.25$ and select $\kappa$ based on how likely the relevance parameters of the true covariates are to reach zero during optimization; see Section~\ref{subsec:sensitivity} for a more detailed analysis.


\subsection{Mini-batching for Vecchia approximation}
\label{subsec:minibatch}

Although the Vecchia approximation has reduced the complexity of model estimation to be linear in $n$, we aim to further improve the computation efficiency of VGPR through mini-batch subsampling that has created considerable success in stochastic gradient descent. In this section, we propose a subsampling method specific to the Vecchia approximation that reduces the complexity to be linear in the batch size $\check{n}$ and leads to unbiased estimating equations. Specifically, we propose to sample the summands of the scaled Vecchia log-likelihood $\hat{\ell}^{\tilde{\br}}(\bftheta)$:
\begin{align}
    \label{equ:mini-batch subsampling}
    \check{\ell}^{\tilde{\br}}(\bftheta) = \sum_{i \in \mathcal{S} \subset \{1, \ldots, n\}} \log p_{\bftheta}(y_i | \by_{c(i)}),
\end{align}
with equal probability and without replacement. Here, $\mathcal{S}$ is the mini-batch index set of size $\check{n}$ and we use $\check{\ell}^{\tilde{\br}}(\bftheta)$ and $\check{h}_{\lambda}^{\tilde{\br}}(\bftheta)$ to denote the counterparts of $\hat{\ell}^{\tilde{\br}}(\bftheta)$ and $h_{\lambda}^{\tilde{\br}}(\bftheta)$ under mini-batch subsampling. 

This mini-batch subsampling can be applied to covariate selection and parameter estimation through slight modifications to Algorithms~\ref{alg:forward_backward} and \ref{alg:QCCD}, respectively, as indicated by their blue underscored components. To avoid oscillation around the optimum, which is a common issue for mini-batch subsampling, we apply the technique introduced in \cite{chee2018convergence} to our QCCD algorithm for the detection of stationarity as indicated in Lines~\ref{stp:step_temper_bgn} to \ref{stp:step_temper_end} of Algorithm~\ref{alg:QCCD}. Specifically, the detection depends on the running sum of the inner product of successive stochastic gradients, and the learning rate $\alpha$ is halved upon detection of convergence; refer to Algorithm~1 of \cite{chee2018convergence} for more details.

One advantage of this mini-batch subsampling based on the Vecchia approximation is having unbiased gradient estimators:
\begin{align}
    \label{equ:unbiased_grad_est}
    E[\nabla \check{\ell}^{\tilde{\br}}(\bftheta)] = \nabla \left( E[\sum_{i = 1}^{n}  \log p_{\bftheta}(y_i | \by_{c(i)}) \delta_{i \in \mathcal{S}}] \right) = \nabla \left( \frac{\check{n}}{n} \sum_{i = 1}^{n}  \log p_{\bftheta}(y_i | \by_{c(i)}) \right) = \frac{\check{n}}{n} \nabla \hat{\ell}^{\tilde{\br}}(\bftheta),
\end{align}
which is generally not the case for other mini-batch subsampling methods used in GP regression such as \citet{chen2020stochastic}. In \eqref{equ:unbiased_grad_est}, the expectation is taken with respect to $\mathcal{S} \subset \{1, \ldots, n\}$. The unbiased property of the mini-batch subsampling is relative to $\nabla \hat{\ell}^{\tilde{\br}}(\bftheta)$ as opposed to $\nabla \ell(\bftheta)$; however, optimizing the Vecchia log-likelihood generally leads to the correct values for $\bftheta$:
\begin{proposition}
    Assuming that $\by$ is a realization of a Gaussian process with zero mean and a covariance structure parameterized by $\bftheta = \bftheta_0$ and that $\hat{\ell}^{\tilde{\br}}(\bftheta)$ is its Vecchia-type log-likelihood, the true parameter value $\bftheta = \bftheta_0$ maximizes the expectation of $\hat{\ell}^{\tilde{\br}}(\bftheta)$ with respect to $\by$: $\bftheta_0 \in \mbox{argmax}_{\bftheta}E[\hat{\ell}^{\tilde{\br}}(\bftheta)]$.
    \label{thm:Vecchia_unbiased}
\end{proposition}
\begin{corollary}
    \label{cor:unbias_est_eq}
    $\nabla \hat{\ell}^{\tilde{\br}}(\bftheta) = \bfzero$ are unbiased estimating equations assuming that $\hat{\ell}^{\tilde{\br}}(\bftheta)$ is first-order differentiable.
\end{corollary}
\noindent The proof of Proposition~\ref{thm:Vecchia_unbiased} is in the Appendix, based on which the proof of Corollary~\ref{cor:unbias_est_eq} is straight-forward. \citet{Stein2004} showed that the Vecchia approximation of the restricted log-likelihood leads to unbiased estimating equations; here, we provide a stronger result for the Vecchia approximation of the log-likelihood. 

We numerically compared our subsampling strategy to two other strategies in terms of the bias and the variance of their gradient estimators. Comparison method I selects $\check{n}$ responses from $\{(y_i)\}_{i = 1}^n$ with equal probability and without replacement, and then the scaled Vecchia approximation for the GP defined over $\{(y_i)\}_{i \in \mathcal{S}}$ is used for computing the SR-gradient. Comparison method II is similar to what we proposed in \eqref{equ:mini-batch subsampling}, sampling the summands of $\hat{\ell}^{\tilde{\br}}(\bftheta)$ but with probabilities proportional to $i^{-1/d}$ and without replacement, which compensates the $\order(i^{-1/d})$  decrease of $\min_{j \in c(i)} \| \bx_i - \bx_j \|$ \citep[e.g.,][]{Katzfuss2021} into consideration and balances the presences of short-range and long-range distances.
\begin{figure}
    \centering
    \begin{subfigure}{.5\textwidth}
        \centering
         \includegraphics[width =.98\linewidth]{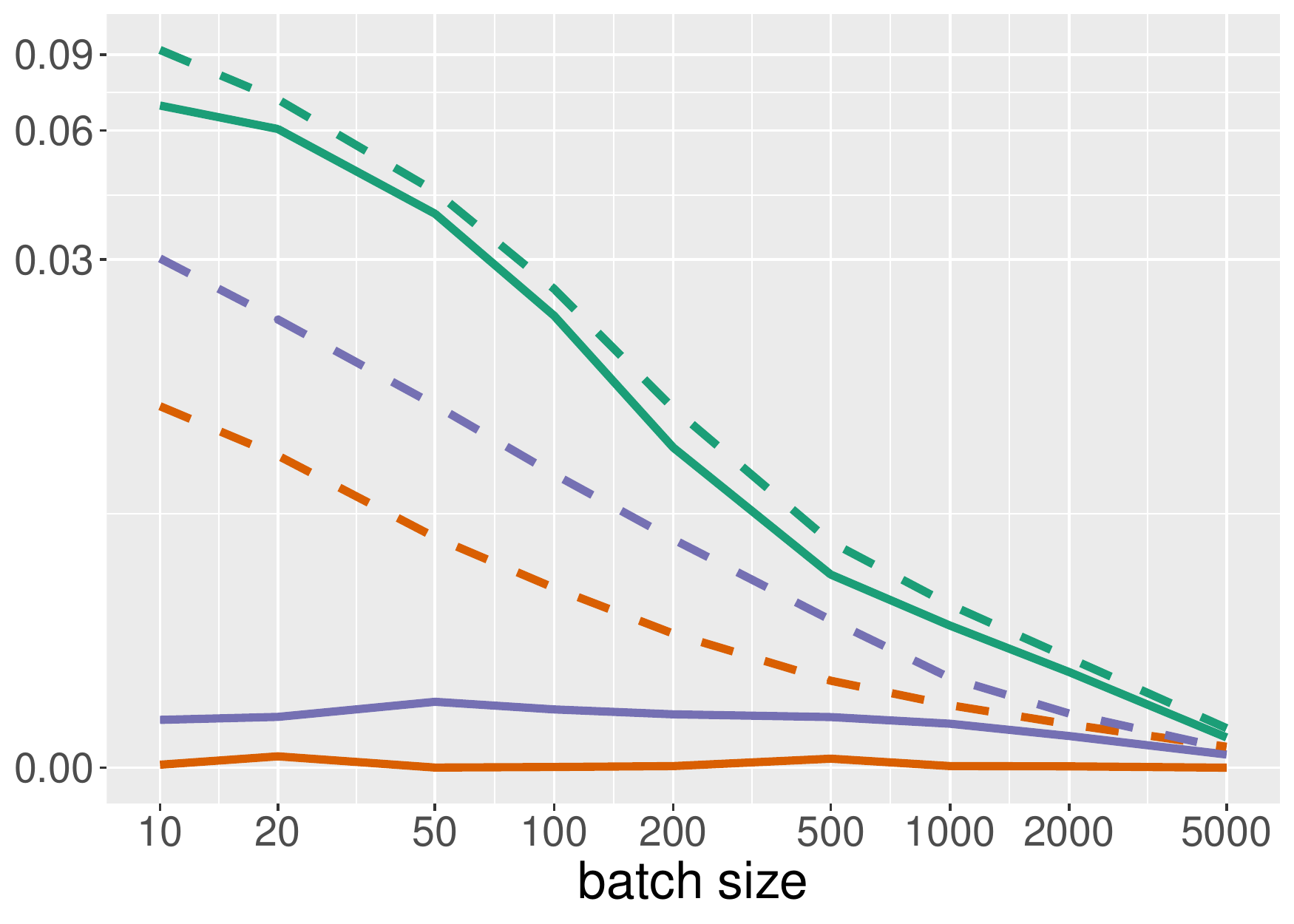}
        \caption{Absolute bias and RMSE of $\frac{\partial \hat{\ell}^{\tilde{\br}}(\bftheta)}{\partial r_1^2}$}
        \label{fig:grad_bias_SR1}
        \end{subfigure}%
        \hfill
    \begin{subfigure}{.5\textwidth}
        \centering
         \includegraphics[width =.98\linewidth]{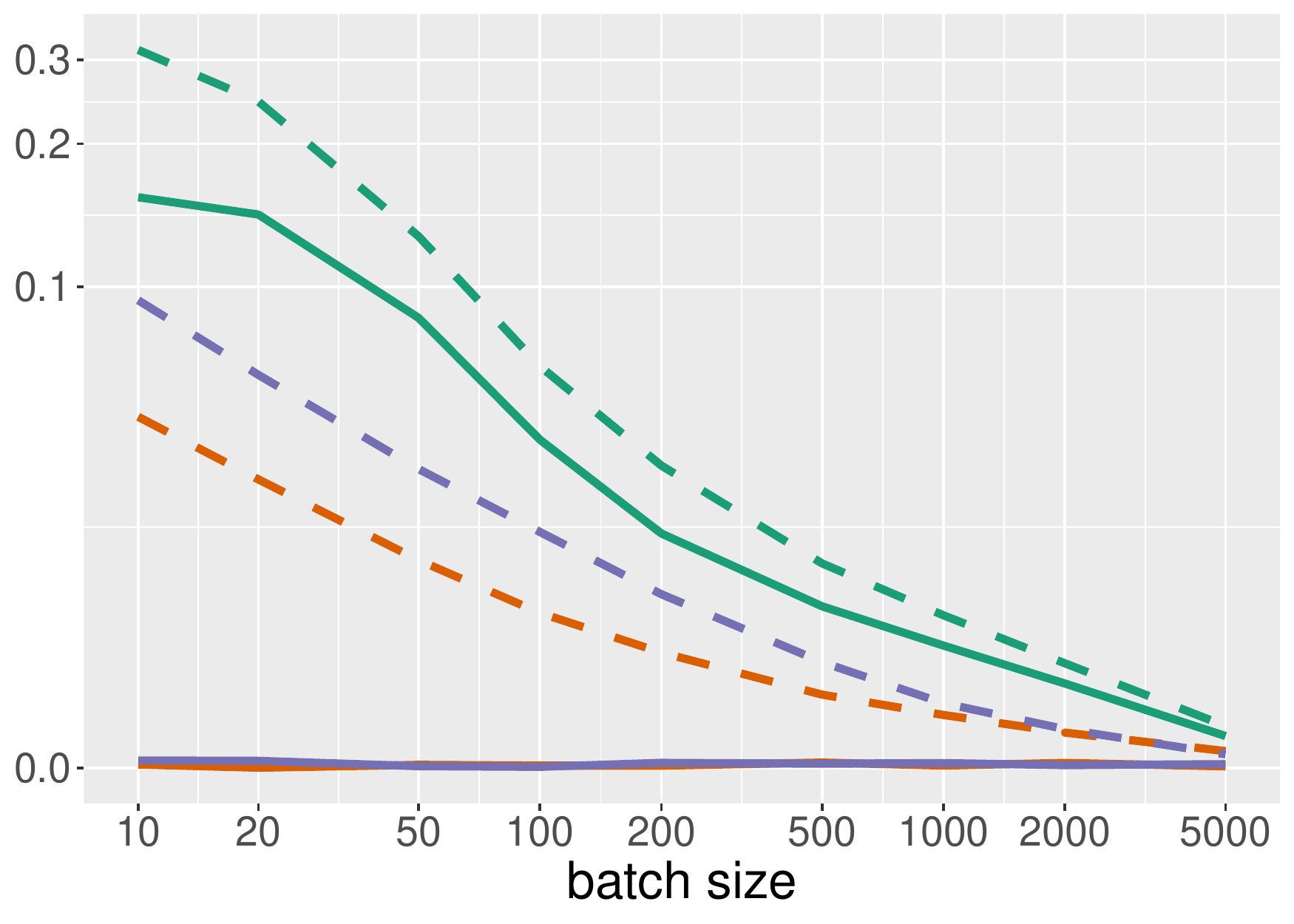}
        \caption{Absolute bias and RMSE of $\frac{\partial \hat{\ell}^{\tilde{\br}}(\bftheta)}{\partial r_2^2}$}
        \label{fig:grad_bias_SR2}	
    \end{subfigure}%
    \caption{Absolute bias (solid) and RMSE (dashed) of the SR-gradient estimators of the three mini-batch subsampling methods. Red, green, and blue represent our proposed mini-batch subsampling of \eqref{equ:mini-batch subsampling}, comparison method I and comparison method II, respectively.}
    \label{fig:cmp_sampling}
\end{figure}
Figure~\ref{fig:cmp_sampling} compares the three mini-batch subsampling methods using a GP defined over $n = 10^4$ locations in $\mathbb{R}^2$ whose true parameters $\bftheta$ are $(\sigma^2, r_1, r_2, \tau^2) = (1, 1, 0.5, 0.05^2)$ and assumed known. The numbers of mini-batches averaged over are $5{,}000$ if the batch size is smaller than $500$ and $500$ otherwise. Our proposed sampling method had the smallest empirical absolute bias and RMSE, highlighting its advantage as the SR-gradient estimator. Comparison method I, as a most intuitive mini-batch subsampling method, leads to a poor gradient estimator because the responses not selected in $\mathcal{S}$ are ignored, losing significant amount of information compared with the other two methods. While it is desirable to reduce the dependence within each mini-batch, the smaller bias and variance of our proposed method over Comparison method II suggests finding sampling probabilities that lead to smaller variance is non-trivial and may lead to nonzero bias.

Figure~\ref{fig:reg_path_minibat} shows regularization paths in the same setting as in Figure~\ref{fig:reg_path}, except using mini-batch subsampling with a batch size of $128$ and increasing $\kappa$ in the penalty function from zero to two. The choice of batch size poses a trade-off between computation efficiency and variability of the gradient estimator, which may depend on the training dataset and computation capacity; in general, larger batches improve the convergence stability but increase the computational cost. A discussion on the choice of $\kappa$ is provided in Section~\ref{subsec:sensitivity}.
\begin{figure}
\centering
	\begin{subfigure}{.5\textwidth}
	\centering
 	\includegraphics[width =.98\linewidth]{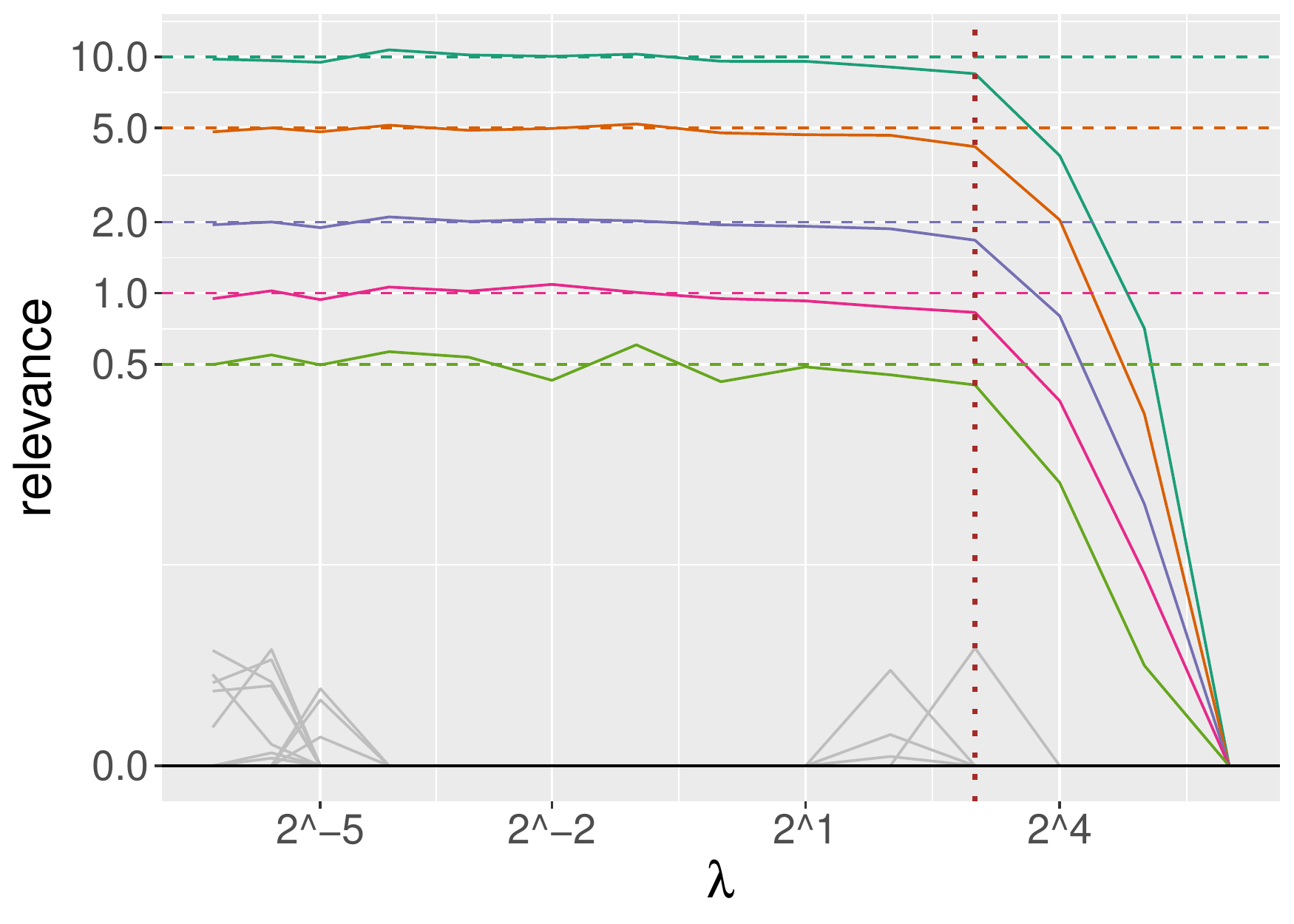}
	\caption{Independent covariates}
	\label{fig:QCCD_bridge_minibat_5000_1000}
	\end{subfigure}%
    \hfill
	\begin{subfigure}{.5\textwidth}
	\centering
 	\includegraphics[width =.98\linewidth]{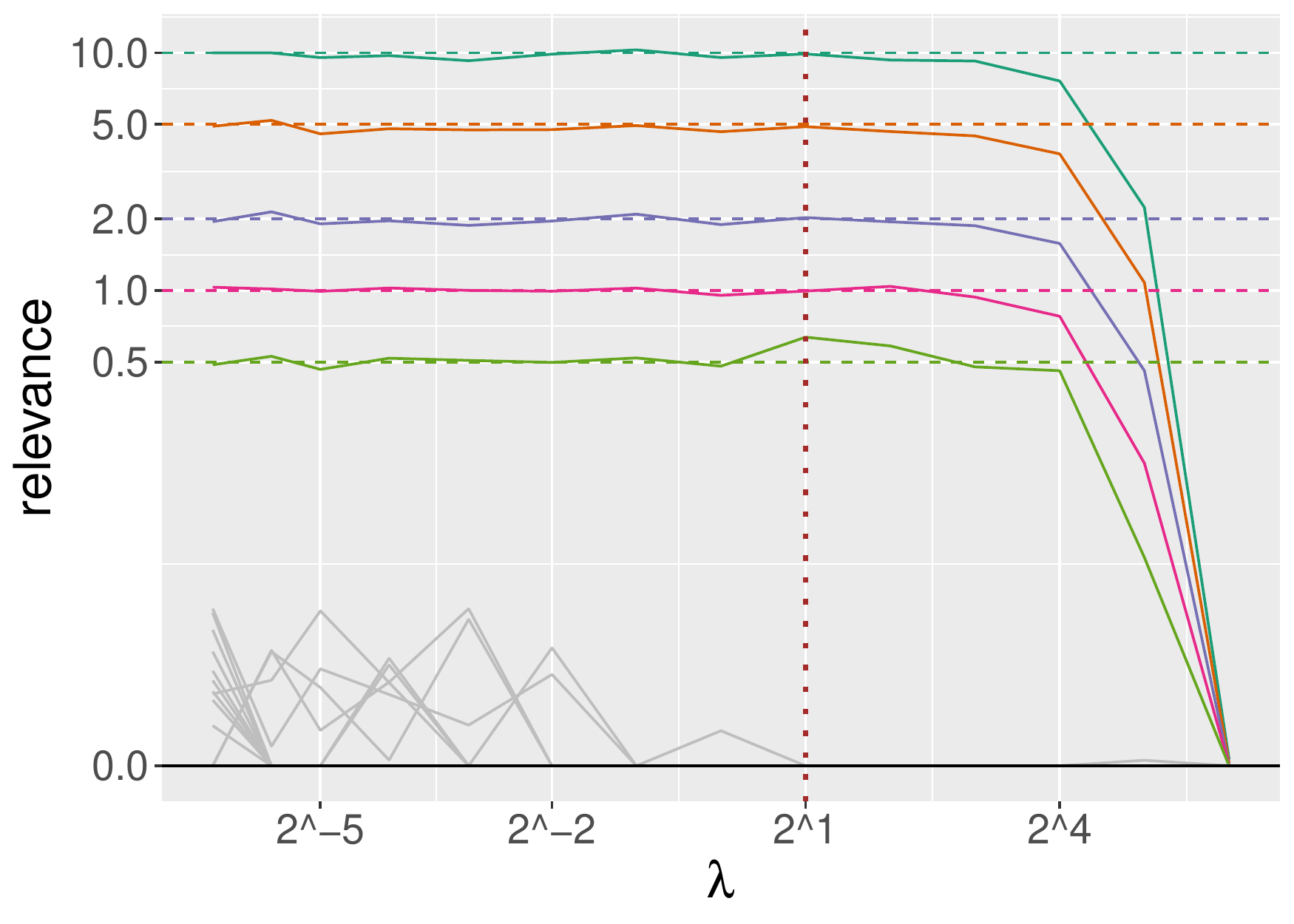}
	\caption{Dependent covariates}
	\label{fig:QCCD_bridge_minibat_dep_5000_1000}
	\end{subfigure}%
	
    \caption{Regularization path computed by VGPR with mini-batch subsampling, using independent or dependent covariates. The relevance parameters of the true covariates are color-coded, and their true values are marked by horizontal colored dashed lines. The fake covariates, whose true relevance parameters are zero, are colored in grey. The vertical red dotted lines mark the optimal model indicated by the stopping condition.}
    \label{fig:reg_path_minibat}
\end{figure}
The estimated models indicated by the red dashed lines were almost the same as those in Figure~\ref{fig:reg_path} while the computation time was reduced by more than $90\%$. When considering the overall sparsity patterns in Figures~\ref{fig:reg_path} and \ref{fig:reg_path_minibat}, the combination of mini-batch subsampling and the iterative adaptive bridge penalty leads to a stronger capacity of deselecting fake covariates, because different mini-batches tend to select the same set of true covariates but different sets of fake covariates, inducing bigger variance on the gradient estimators of the fake covariates while $c_{\iota, l}^{\kappa}$ in \eqref{equ:adbridge} is smaller for the fake covariates, indicating stronger penalization. 

\subsection{Complexity analysis}
\label{subsec:complexity}

In this section, we analyze the computation gains from using the Vecchia approximation, the VGPR algorithm introduced in Algorithm~\ref{alg:VGPR}, and the mini-batching technique from Section~\ref{subsec:minibatch}. The Vecchia approximation reduces the complexity of computing the log-likelihood and its gradient from $\order(n^3)$ and $\order(n^3d)$ to $\order(nm^3)$ and $\order(ndm^3)$, respectively; refer to \cite{Guinness2019} for the gradient computation under the Vecchia approximation. Based on the intermediate results from the gradient computation, the FIM of the Vecchia log-likelihood needs only $\order(nd^2m^2)$ additional operations. The VGPR algorithm reduces the number of covariates involved in optimization, reducing the $d$ in aforementioned complexities to $|\zeta|$, with $|\zeta| \approx d_0 \ll d$. Finally, the mini-batching technique further reduces the $\order(n)$ complexities to $\order(\check{n})$, leading to $\order(\check{n}m^3)$, $\order(\check{n}d_0m^3)$, and $\order(\check{n}d_0^2m^2)$ complexities for computing the objective function, its gradient, and FIM, respectively. The SR-gradient of all $d$ covariates is needed in Algorithm~\ref{alg:forward_backward} to select $k$ new covariates at the cost of $\order(\check{n}dm^3)$, but its computation frequency is negligible compared with the number of gradient computations needed by QCCD and it is typically a minor component in the overall computation cost. 

For very large $n$, it is also possible to reduce the cost of MM and NN by replacing them by random ordering and the index-based-on-inverted-file (IVF) method \citep[implemented in the Faiss library of][]{JDH17}, respectively. The cost of NN could be further reduced by computing the $m$ nearest neighbors on-the-fly only for the responses in the mini-batch $\mathcal{S}$.

GP prediction also benefits significantly from the techniques introduced in Section~\ref{sec:GPR_var_select}. The nearest neighbors of each test point can be computed much faster in $d_0$ dimensions than in $d$ dimensions, based on which posterior inference at each test point can be achieved in $\order(m^3 + d_0 m^2)$ time using the scaled Vecchia approximation, assuming that the number of selected covariates is $\order(d_0)$.

\subsection{Sensitivity to tuning parameters}
\label{subsec:sensitivity}

The VGPR algorithm includes several tuning parameters that are considered fixed when running the algorithm. We provide some guidance here.
Larger values of the conditioning set size $m$ lead to more accurate approximation of the exact GP and we choose $m = 100$ based on \citet{Katzfuss2020} and for computational feasibility. 
The Armijo constant $c$ in Algorithm~\ref{alg:QCCD} heuristically prevents `overly large' steps; we choose $c = 10^{-4}$ as recommended in Chapter~3 of \cite{wright1999numerical} and used in the \texttt{GpGp} R package \citet{Guinness2019}. 
The learning rate parameter $\alpha$ in Algorithm~\ref{alg:CCD} reduces oscillation around an optimum, hence promoting convergence; we use the same initialization (i.e., $\alpha = 1$) and scaling (i.e., by $1/2$) for $\alpha$ as in \citet{chee2018convergence}, where this oscillation-reduction technique was proposed.  

The number of new covariates selected each iteration ($k$) and the penalty parameters $\kappa$ and $\gamma$ in \eqref{equ:adbridge} are unique to our proposed VGPR algorithm and iterative adaptive bridge penalty, and hence they have not been discussed in the existing literature. Here, we provide some recommendations and a sensitivity analysis on them; see Appendix~\ref{app:sensitivity} for more details. Larger $k$ leads to higher optimization efficiency but also the risk of local optima; we recommend a value between $3$ and $5$. Bigger $\kappa$ corresponds to weaker numerical singularity at $r_l = 0$. We recommend $\kappa > 0$ when mini-batch subsampling is applied and a large $\kappa$ (e.g., 10 or 15) when the GP with ARD kernels is likely a misspecified model. Smaller $\gamma$ causes higher difference in the penalty derivatives at small and large $r_l$. \cite{yi2011penalized} used $\gamma = 0.01$, whereas we recommend a choice between $0.1$ and $0.25$ for a smoother objective function. Based on Appendix~\ref{app:sensitivity}, we conclude that the VGPR algorithm is overall not sensitive to the choice of $k$, $\kappa$, and $\gamma$.

\section{Simulation Study}
\label{sec:sim_study}

\subsection{Simulation setup}

We compared the VGPR algorithm proposed in Algorithm~\ref{alg:VGPR} with methods commonly used in machine learning for variable selection or GP model estimation, namely Lasso regression \citep{Tibshirani1996}, the sparse additive model \citep[SAM;][]{ravikumar2009sparse}, regression trees \citep[Tree;][]{loh2011classification}, penalized GP regression \citep[PGPR][]{yi2011penalized}, kernel interpolation for scalable structured Gaussian processes \citep[KISS;][]{wilson2015kernel}, Vecchia Fisher scoring \citep[Fisher;][]{Guinness2019}, and GPs with forward selection (FWD). We used the scaled Vecchia approximation in Fisher and FWD but the exact GP log-likelihood in PGPR to respect the original algorithm of \citet{yi2011penalized}. For `Tree' and `Lasso', the default setups from the `glmnet' R package \citep{friedman2010regularization} and the `sklearn' Python module \citep{scikit-learn} were used, respectively. KISS generally has high scalability in $n$ but low scalability in $d$. Based on the GPyTorch \cite{gardner2018gpytorch} implementation, when $d > 5$, the kernel function needed to assume an additive structure to be computationally feasible, for which ARD kernels are yet available, hence we only consider KISS as a state-of-the-art competitor for prediction at unknown locations. We generated $d$ independent or dependent covariates at $(n + 5000)$ locations and simulated $(n + 5000)$ GP responses. $5{,}000$ responses were set aside as the testing dataset used to evaluate the four methods' performances. We considered $n \in \{500, 5{,}000, 25{,}000\}$, $d \in \{ 100, 1{,}000\}$, and independent versus dependent covariates, for a total of $12$ simulation scenarios. Methods were compared from three aspects, namely posterior prediction as measured by the RMSE based on the test dataset, misclassification ratios as measured by false positive and false negative ratios, and computation times. 

`PGPR' and `VGPR' use penalty functions, for which we chose the classic bridge penalty and the iterative adaptive bridge penalty as in Section~\ref{subsec:minibatch}, respectively, to compute their regularization paths. Methods involving solution paths, including `VGPR', needed an OOS score in their stopping conditions, for which a quarter or $5{,}000$, whichever is smaller, of the training dataset was set aside and only used in computing the OOS RMSE. Similar to Sections~\ref{subsec:traverse_path} and \ref{subsec:minibatch}, all stopping conditions were defined as producing less than $1\%$ improvement of OOS RMSE after the selection of any new covariate. The OOS RMSE was also used to choose the best model in each iteration of forward selection. Fisher scoring does not require a stopping condition based on an OOS score and hence used the whole training dataset for parameter estimation.

Because Fisher scoring and the conjugate gradient used in \cite{yi2011penalized} are unconstrained optimization algorithms that rely on variable transformations, their parameters, including $\br$, cannot reach exact zeros. We set a cut-off threshold of $10^{-7}$, the same as in \citet{yi2011penalized}, below which the corresponding covariate was viewed as deselected. The initial values for $\sigma^2$, $\{r_l\}_{l = 0}^{d}$, and $\tau^2$, when needed, were $0.25$, $0.1$, and $10^{-4}$, respectively, while for PGPR, ten random initial values, as recommended in \cite{yi2011penalized}, were used for the optimization at each $\lambda$. The maximum numbers of iterations were $100$ for PGPR, Fisher, and FWD, while $200$ for VGPR, as the latter used mini-batch subsampling with $\check{n} = 128 \ll n$. The computation times were measured on an Intel Xeon E5-2680 v4 CPU using 56 cores and capped at a 10-hour limit for each GP replicate.

\subsection{Simulation results}

The comparison results are shown in Figure~\ref{fig:sim_metrics}.
\begin{figure}
    \centering
 	\includegraphics[width =.98\textwidth]{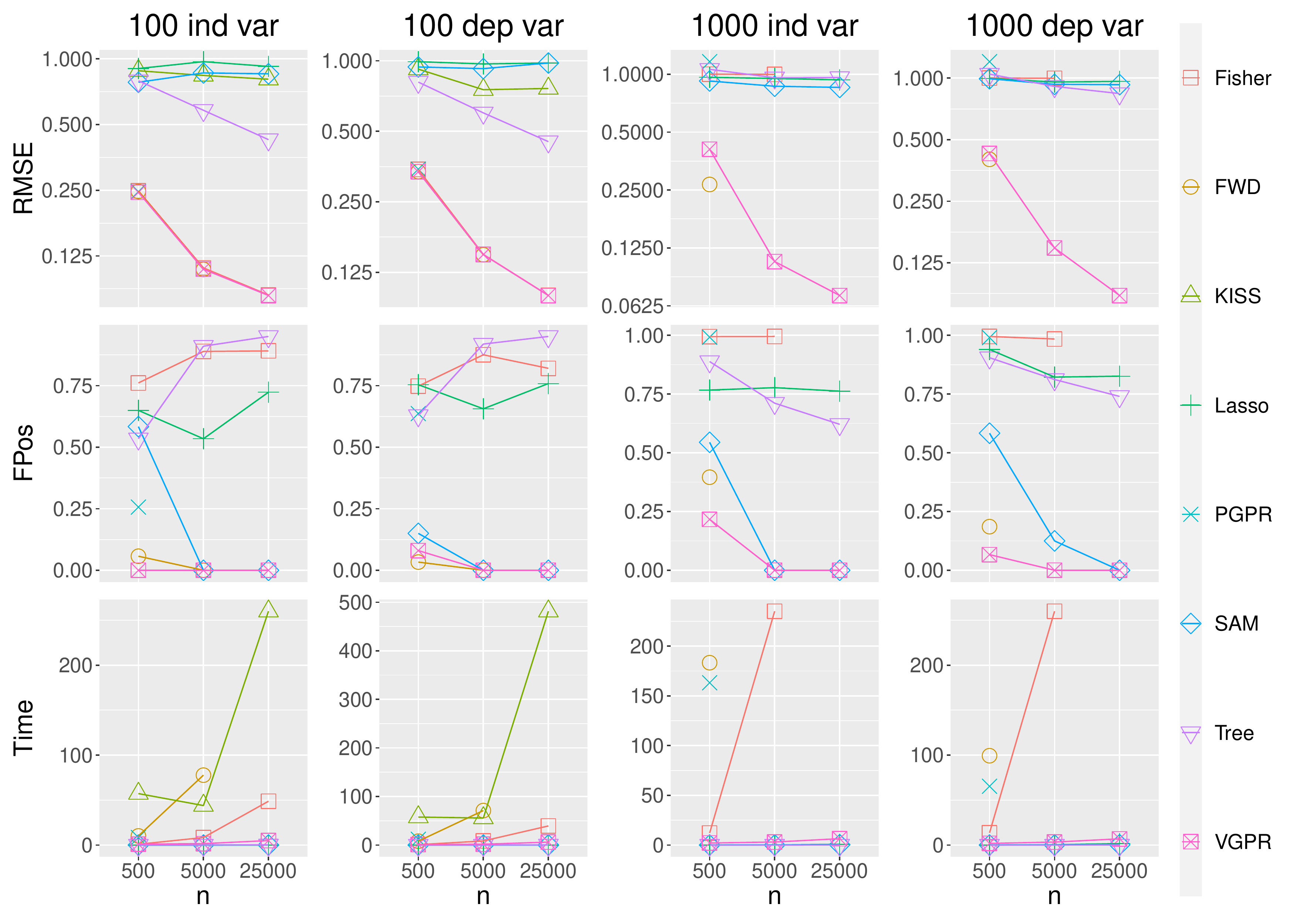}
	
    \caption{Comparison of eight methods for variable selection and/or GP regression, in terms of RMSE, false positive rates (FPos), and computation time in minutes. The results were averaged over five replicates. When $d = 100$, `Fisher', `FWD', `PGPR', and `VGPR' had close RMSE scores when available. The computation times of `Lasso', `SAM', and `Tree' are similar, all faster than `VGPR'.}
    \label{fig:sim_metrics}
\end{figure}
The RMSEs of `Fisher', `FWD', `PGPR', and `VGPR', when computationally feasible, were similar for $d = 100$ but diverged for $d = 1{,}000$, indicating convergence to local optima when the number of optimization parameters was high. Specifically, both `PGPR' and `Fisher' involve $\order(d)$ parameters in optimization, while `FWD' and `VGPR' sequentially increase the number of parameters based on warm starts, which achieved significantly better result for reaching the global optimum. While `FWD' provided slightly more accurate predictions than `VGPR' for $n = 500$, it quickly became computationally infeasible as $n$ or $d$ increased. In contrast, `VGPR' had a better tradeoff between data efficiency and computation scalability. `Lasso', `SAM', and `Tree' were less suitable for the simulated multivariate normal datasets due to model misspecification. While `KISS' is a GP-based model, its idea of finding a (large) common set of pseudo-inputs for all locations became impractical when $d$ is moderately large. In terms of `FPos', which measures the proportion of fake covariates among the selected, `VGPR' outperformed all other methods, achieving zero `FPos' ratios when $n \ge 5{,}000$. This highlights the capability of `VGPR' for deselecting fake covariates, hence the advantages of using QCCD over conjugate gradient and Fisher scoring for simultaneous variable selection and parameter estimation. The false negative ratios were almost constantly zero for all methods, and are hence not shown. Although slightly slower than the compared machine-learning models, `VGPR' tremendously outperformed the other GP-based methods, becoming the only feasible GP-based method when $n = 25{,}000$ and $d = 1{,}000$ under the 10-hour limit.

\section{Application Study}
\label{sec:app_study}

We performed a comparison on several real datasets and data produced by a physical model. Specifically, we compared the methods from Section~\ref{sec:sim_study} that are computationally feasible at $n = 25{,}000$ and $d = 1{,}000$, namely `Lasso', `Tree', `VGPR', and `SAM'. For these examples, our assumed model, GP with ARD covariance kernels, is likely to be misspecified and furthermore, true covariates may not exist in the given covariate pool, and so we used the iterative adaptive bridge penalty with $\kappa$ bigger than those in Sections~\ref{sec:GPR_var_select} and \ref{sec:sim_study} to select the most predictive covariates.

The first dataset was generated from the Piston function \citep[e.g.,][]{Surjanovic2013}, which is a (deterministic) physical model with $d_0 = 7$ true covariates; a total of $d = 10^3$ covariates were simulated at $n = 10^6$ locations. While the underlying model is not a GP, the true covariates were included in the covariate pool, and we chose $\kappa = 5$. The second dataset was the ``relative location of CT slices on axial axis'' (Slice) from the UCI Machine Learning Repository \citep{Dua2019} that has $n = 53{,}500$ images and $d = 386$ features. The images belong to $74$ individuals, among which a quarter were selected as the testing dataset. The third dataset was the ``physicochemical properties of protein tertiary structure'' (CASP) dataset also from the UCI Repository with $n = 45{,}730$ responses and $d = 9$ features. A fourth dataset was the temperature (Temp) data used in \cite{garnett2013active} that contains $7{,}117$ training samples and $3{,}558$ testing samples, each with $d = 106$ features. For the last three datasets, we set $\kappa = 15$ to compensate for model misspecification and the potential lack of true covariates, and supplemented with $(10^3 - d)$ artificial covariates such that a total of $10^3$ covariates were used to compare the three methods' capability of variable selection. Similar to previous experiments, we generated either uncorrelated or correlated covariates but here, the latter was constructed as random linear combinations of original covariates plus independent Gaussian noise. Both the responses $\by$ and the covariates $\bx_i$ were standardized to have zero mean and unit variance. 

Although the set of true covariates was unknown, misclassification ratios, specifically the false positive ratio, could still be estimated based on the number of included artificial covariates. On the other hand, the number of selected covariates is also an important indicator for the quality of variable selection that directly relates to over-fitting and computation efficiency. Table~\ref{tbl:real_data_app} summarizes the three metrics of the four methods under the previously mentioned datasets. 
\begin{table}[h]
    \centering
    \begin{tabular}{l | l | r | r | r || l | l | r | r | r}
        \small
        Dataset & Method & RMSE & nSel & FPos & Dataset & Method & RMSE & nSel & FPos \\
        \hline Piston-I & VGPR & 0.00 & 7 & 0\% & 
         CASP-I & VGPR & 0.78 & 3 & 0\%\\
         & Lasso & 0.17 & 7 & 0\% & & Lasso & 0.85 & 177 & 96\%\\
         & SAM & NA & NA & NA & & SAM & 0.84 & 6 & 0\%\\
         & Tree & 0.92 & 6 & 17\% & & Tree & 0.92 & 6 & 17\%\\
         \hline Piston-D & VGPR & 0.00 & 7 & 0\% & 
         CASP-D & VGPR & 0.75 & 10 & 60\%\\
         & Lasso & 0.17 & 7 & 0\% & & Lasso & 0.85 & 221 & 96\%\\
         & SAM & NA & NA & NA & & SAM & 0.84 & 6 & 0\%\\
         & Tree & 0.71 & 147 & 95\% & & Tree & 0.92 & 6 & 33\%\\
         \hline Slice & VGPR & 0.38 & 64 &  & 
         Temp & VGPR & 0.29 & 6 & \\
         & Lasso & 0.44 & 359 &  & & Lasso & 0.28 & 84 &\\
         & SAM & 0.50 & 118 & & & SAM & 0.29 & 40 &\\
         & Tree & 0.44 & 334 &  & & Tree & 0.32 & 26 & \\
         \hline Slice-I & VGPR & 0.32 & 49 & 18\% & 
         Temp-I & VGPR & 0.29 & 8 & 0\%\\
         & Lasso & 0.44 & 792 & 59\% & & Lasso & 0.29 & 122 & 79\%\\
         & SAM & 0.50 & 118 & 0\% & & SAM & 0.29 & 36 & 8\%\\
         & Tree & 0.51 & 294 & 31\% & & Tree  & 0.40 & 47 & 49\%\\
         \hline Slice-D & VGPR & 0.31 & 50 & 12\%& 
         Temp-D & VGPR & 0.29 & 9 & 0\%\\
         & Lasso & 0.43 & 696 & 54\% & & Lasso & 0.29 & 92 & 70\%\\
         & SAM & 0.49 & 142 & 20\% & & SAM & 0.29 & 36 & 22\%\\
         & Tree & 0.41 & 908 & 64\% & & Tree & 0.39 & 46 & 46\%\\
    \end{tabular}
    \caption{Performance comparison of Lasso linear regression (Lasso), sparse additive model (SAM), regression tree (Tree), and VGPR. `Slice', `Piston', `CASP', and `Temp' are dataset names. `I' and `D' indicates being supplemented by uncorrelated and correlated artificial covariates, respectively. `RMSE' measures the RMSE based on the testing dataset. `nSelect' is the number of selected covariates. `FPos' is short for false positive ratio.}
    \label{tbl:real_data_app}
\end{table}
The `Piston' and the `CASP' datasets had too few original covariates to be used for comparing variable selection, and so corresponding results are not listed. The `SAM' method exceeded our memory capacity (128 GB) when $n = 10^6$ using the `Piston' dataset, and so the results are not available. The optimization setups for `VGPR' were the same as in Section~\ref{sec:sim_study}, except for the change of $\kappa$ and that $k$ was increased from $3$ to $5$ to further improve computation efficiency.

`VGPR' outperformed the other three methods in almost all three aspects (same as in Section \ref{sec:sim_study}), especially in terms of the number of selected covariates and the false positive ratios, highlighting the strength of using the iterative adaptive bridge penalty and QCCD for covariate deselection. For the `Piston' dataset, our GP properly captured its non-linear and continuous features, hence predicting with significantly higher accuracy. `VGPR' had a relatively high false positive ratio when the `CASP-D' dataset was used but considering that there were only nine original covariates, the `FPos' was already high with few fake covariates selected. Besides, the fake covariates in this case were correlated with the original covariates, potentially improving posterior inference as reflected by the lower RMSE of `VGPR'. `Lasso' and `SAM' had comparable RMSE to `VGPR' in modeling the `Temp' dataset but its number of selected covariates and `FPos' were significantly higher. The complexity of `VGPR' is tremendously reduced by the Vecchia approximation, gradient-based covariate selection, and mini-batch subsampling, to achieve a computation time of less than forty minutes for a dataset with $n = 10^6$ and $d = 10^3$, for which `Lasso' and `Tree' used sixteen and eight minutes, respectively. Despite being slower, `VGPR' is arguably as scalable as the other two methods (and much more so than existing GP regression methods) based on the complexity analysis in Section~\ref{subsec:complexity}.

\section{Conclusions}
\label{sec:conclusion}

We provide a highly scalable method, coined VGPR, for variable selection and model estimation in GP regression, suitable for datasets with large numbers of responses $n$ and covariates $d$. ARD covariance kernels naturally combine variable selection and model estimation, while a (scaled) Vecchia approximation provides fast and highly parallel computation of the loglikelihood, its gradient and its FIM. We introduced a forward-backward-selection algorithm that iteratively adds predictive covariates to a candidate set $\zeta$ based on the gradient and removes irrelevant covariates from the candidate set using an efficient QCCD algorithm. We provided theoretical support for the gradient-based covariate-candidate selection. To further speed up our method for even larger $n$, we introduced a mini-batch subsampling method specific to Vecchia-type approximations that has unbiased gradient estimators whose expectations are shown to be zero at the true parameter values. The resulting procedure requires only $\order(\check{n}|\zeta|^2 + \check{n}d)$ time, where $\check{n}$ is the mini-batch size, and hence the computational complexity is essentially independent of $n$. To compensate for the sampling variance of the stochastic gradient estimators under mini-batch subsampling, we also introduced an iterative adaptive bridge penalty.

In our simulation study, VGPR was substantially faster and selected fewer (almost zero) false covariates than other state-of-the-art GP regression methods that can be adapted for variable selection. When using real datasets, VGPR was robust enough to select only a small number of the most predictive covariates, maintaining the lowest misclassification ratios and the best predictive power among standard methods for regression with variable selection. VGPR is able to handle $n = 10^6$ responses with $d = 10^3$ features within $40$ minutes on a standard scientific workstation.
Due to its flexibility and accurate results, we consider VGPR to be a suitable candidate for a default benchmark method for nonlinear regression and variable selection on large datasets.

One possible extension of the results in this paper is variable selection and model estimation for generalized GP models, such as logistic or probit GPs for classification problems. For example, \cite{cao2022scalable} derived the marginal and posterior predictive probabilities of the probit GP. A second idea is to examine if the gradient of the objective function or similarly simple criteria can be used to select new covariates for other regression models, hence achieving a forward selection procedure that tremendously benefits the optimization. 

\footnotesize
\section*{Acknowledgments}

Jian Cao was partially supported by the Texas A\&M Institute of Data Science (TAMIDS) Postdoctoral Project program, Jian Cao and Matthias Katzfuss by National Science Foundation (NSF) Grant DMS--1654083, Matthias Katzfuss and Joe Guinness by NSF Grant DMS--1953005, Matthias Katzfuss by NSF Grant CCF--1934904, and Jian Cao and Marc Genton were partially supported by the King Abdullah University of Science and Technology (KAUST). We would like to thank Felix Jimenez for helpful comments and discussions.

\begin{appendices}
\section{Implementation of FIC, FITC, and PIC \label{app:pic}}

FIC selects the first $m$ locations in MM as the inducing inputs. FITC selects the same locations as the initial values of the $m$ inducing inputs, which is then optimized using the `GPflow' Python package, whose result is used as the final inducing inputs of FITC. PIC selects the first $m / 2$ locations in MM as inducing inputs and divides the responses into disjoint subsets of size $m / 2$. PIC considers the subsets of responses conditionally independent given the inducing inputs as opposed to that responses are conditionally independent, which is assumed by FIC and FITC. In other words, PIC considers also local correlation.

Among the five GP approximations, namely, FIC, FITC, PIC, Vecchia and scaled Vecchia, FIC typically has the lowest cost per likelihood estimation, requiring only $\order(nm^2)$ operations because the conditioning sets remain the same for all responses. PIC has higher computation cost than FIC but its complexity stays at the same level. Given the inducing inputs, FITC is as efficient as FIC but the inducing inputs of FITC require an optimization with $\order(md)$ parameters, which could become the dominant complexity. Vecchia and scaled Vecchia approximations have a complexity of $\order(nm^3)$ for likelihood estimation, which although higher than FIC and PIC, is still linear with $n$ and has a highly parallel implementation. Furthermore, the grouping technique introduced in \cite{Guinness2016a} can reduce the previous complexity to between $\order(nm^2)$ and $\order(nm^3)$ and is already implemented in the `GpGp' R package.

\section{Sensitivity Analysis\label{app:sensitivity}}

We generated $d = 100$ dependent covariates at $10^4$ locations, half of which were used for training and the other half were used to compute the RMSE score. The sensitivity was assessed in terms of RMSE score and number of fake covariates selected. From  Figure~\ref{fig:sensitivity}, we conclude that our proposed VGPR algorithm is largely robust across different values of $k$, $\gamma$, and $\kappa$ within the recommended intervals.
\begin{figure}[h]
  \centering
  \begin{subfigure}{.98\textwidth}
      \includegraphics[width =.5\linewidth]{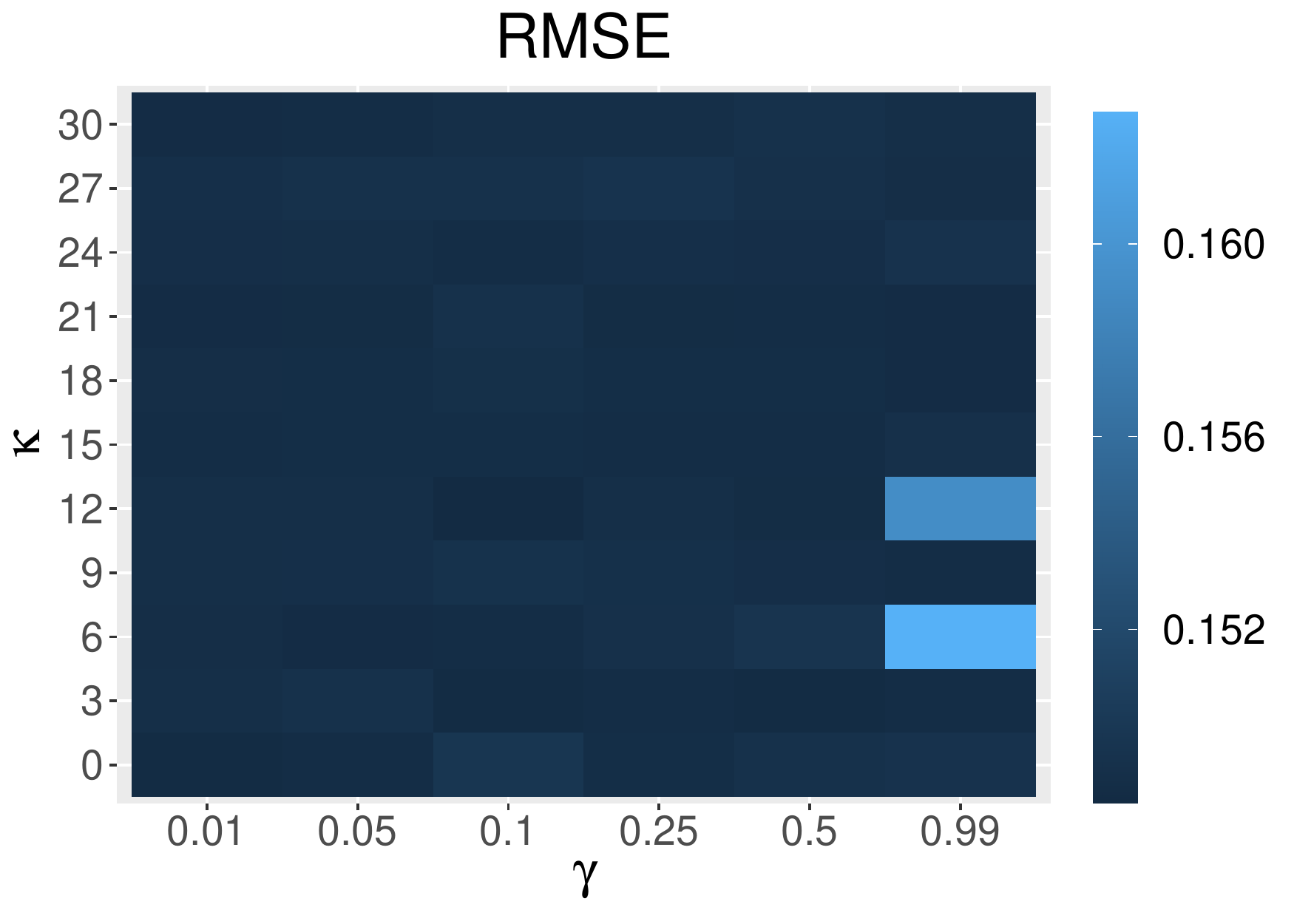}
      \hfill
      \includegraphics[width =.5\linewidth]{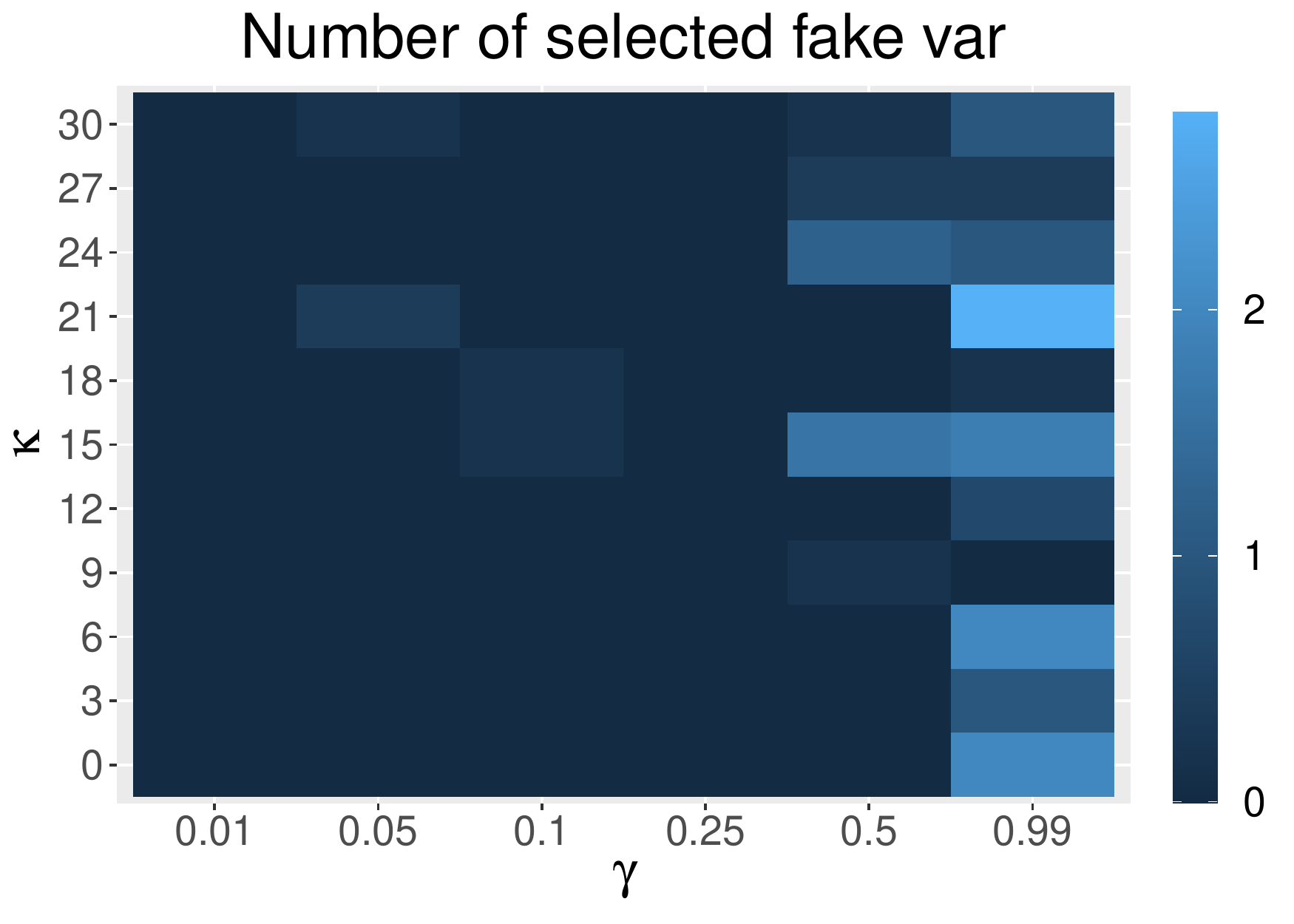}
      \caption{$k = 3$}
  \end{subfigure}
  \begin{subfigure}{.98\textwidth}
      \includegraphics[width =.5\linewidth]{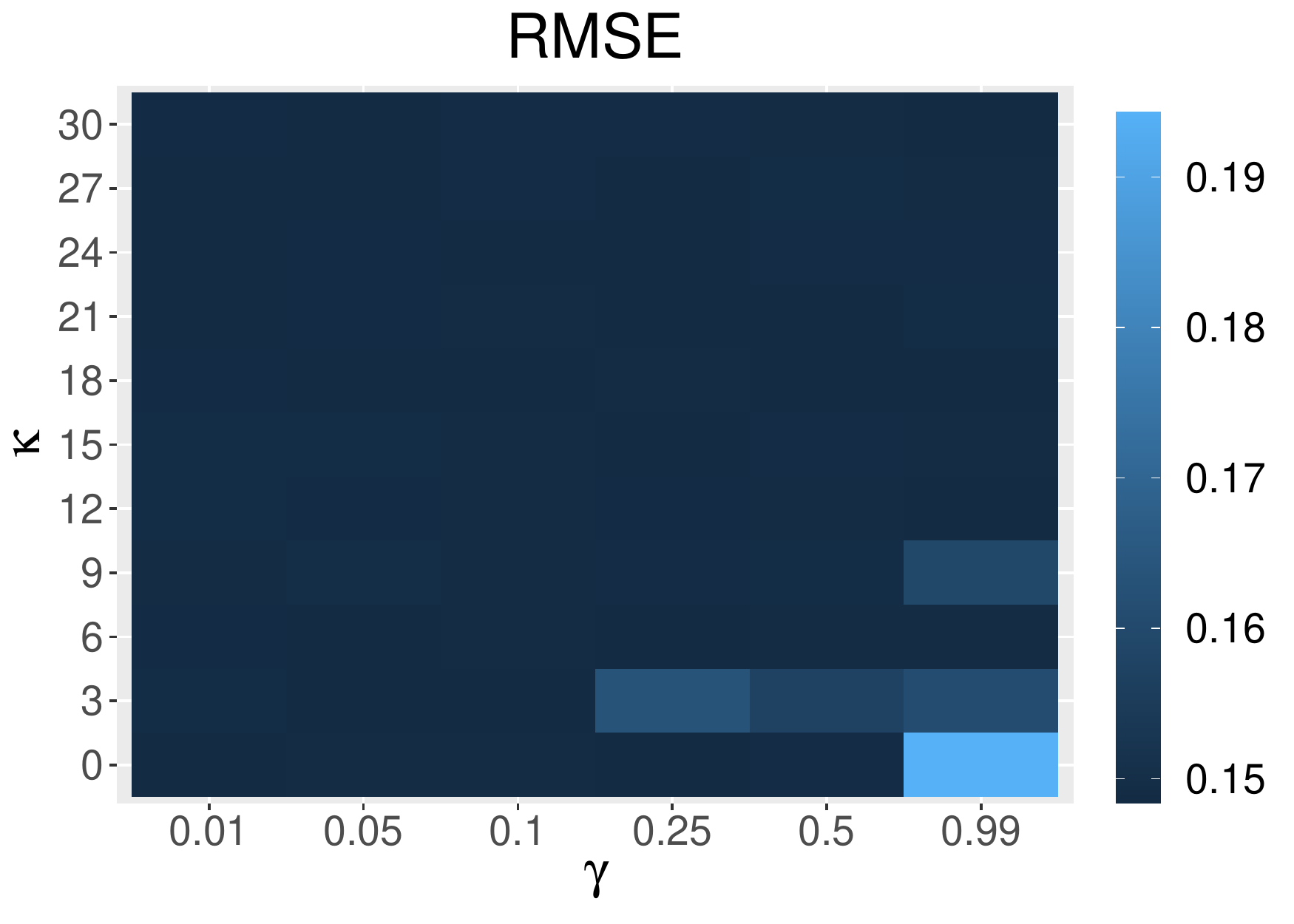}
      \hfill
      \includegraphics[width =.5\linewidth]{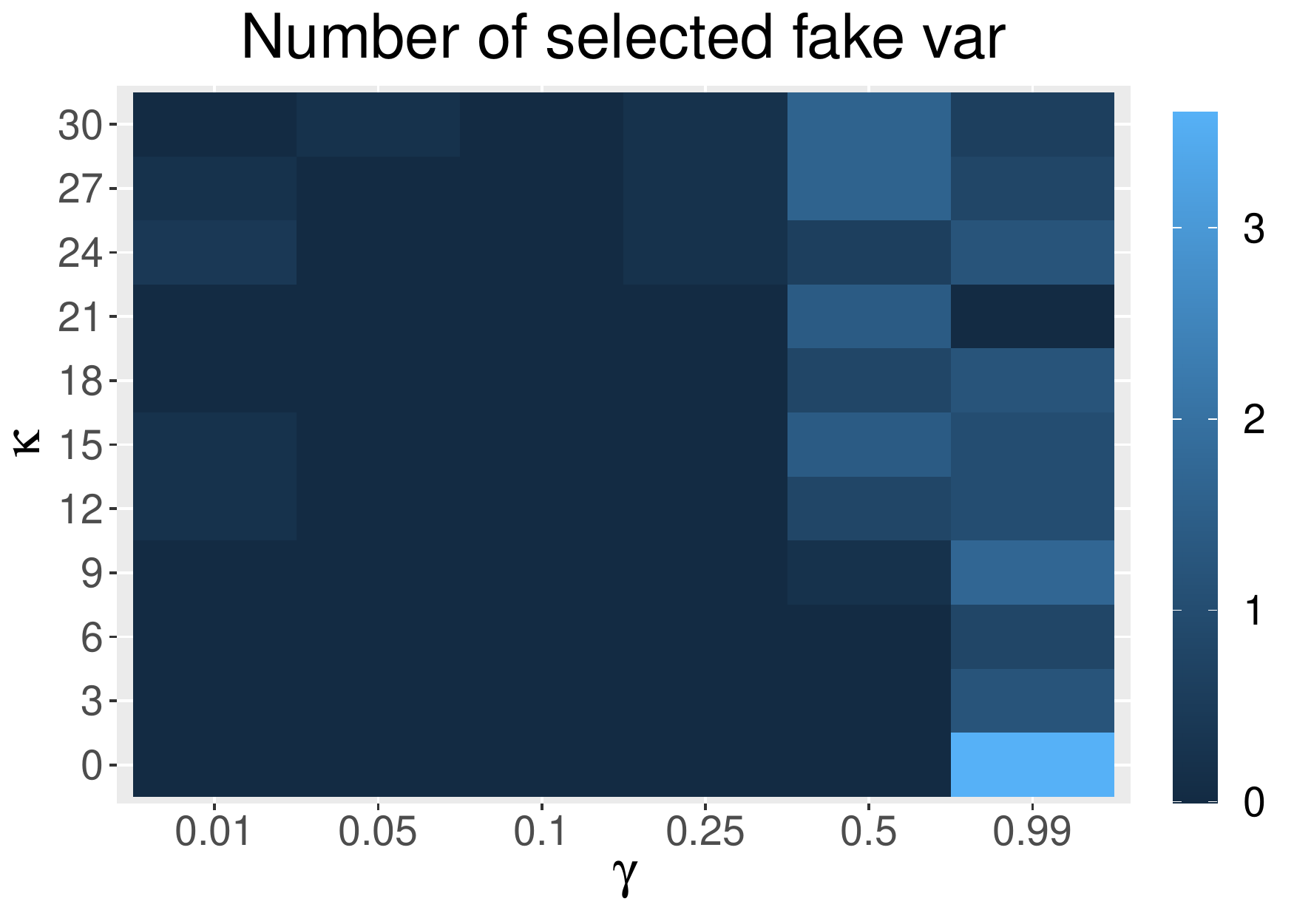}
      \caption{$k = 4$}
  \end{subfigure}
  \begin{subfigure}{.98\textwidth}
      \includegraphics[width =.5\linewidth]{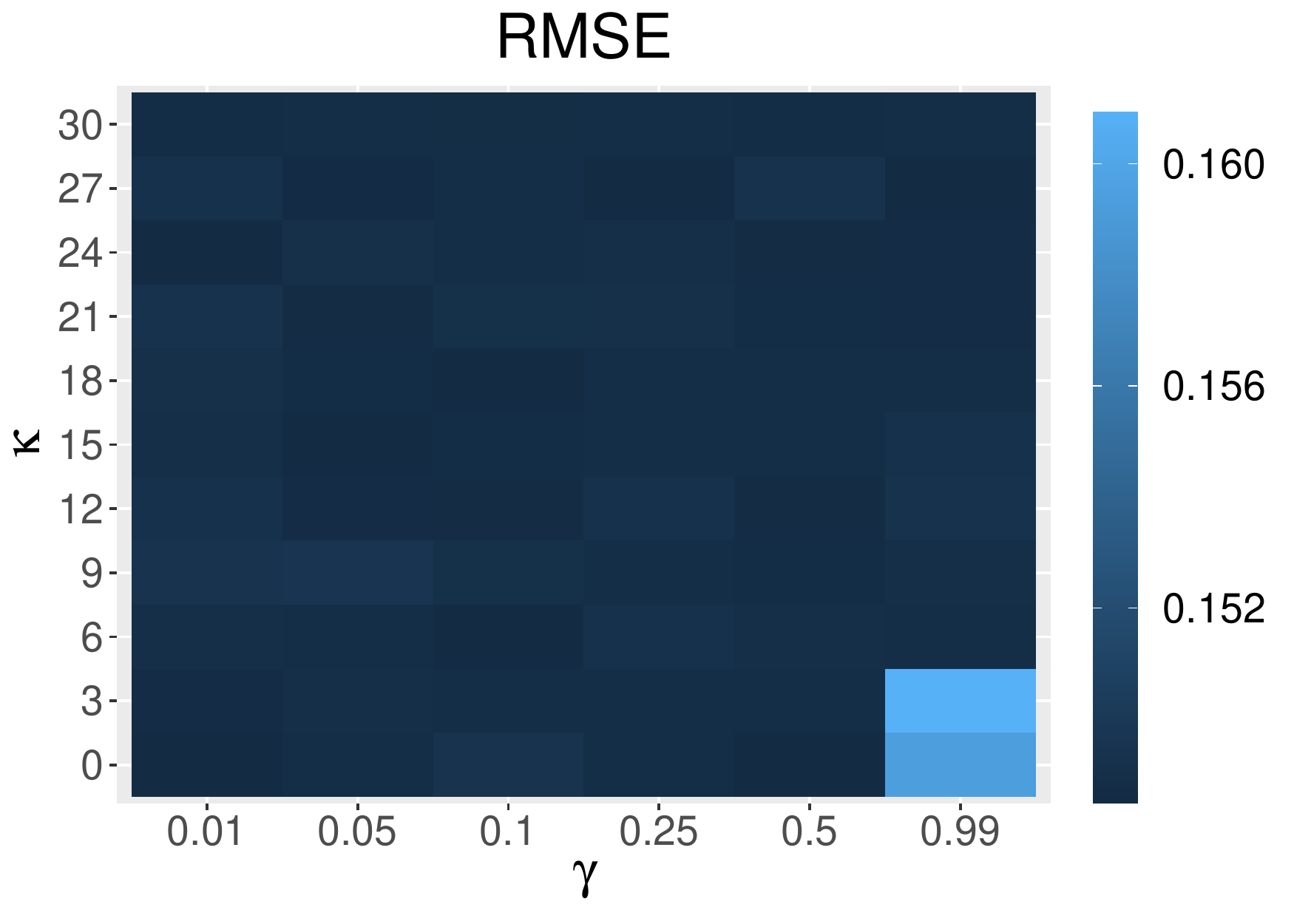}
      \hfill
      \includegraphics[width =.5\linewidth]{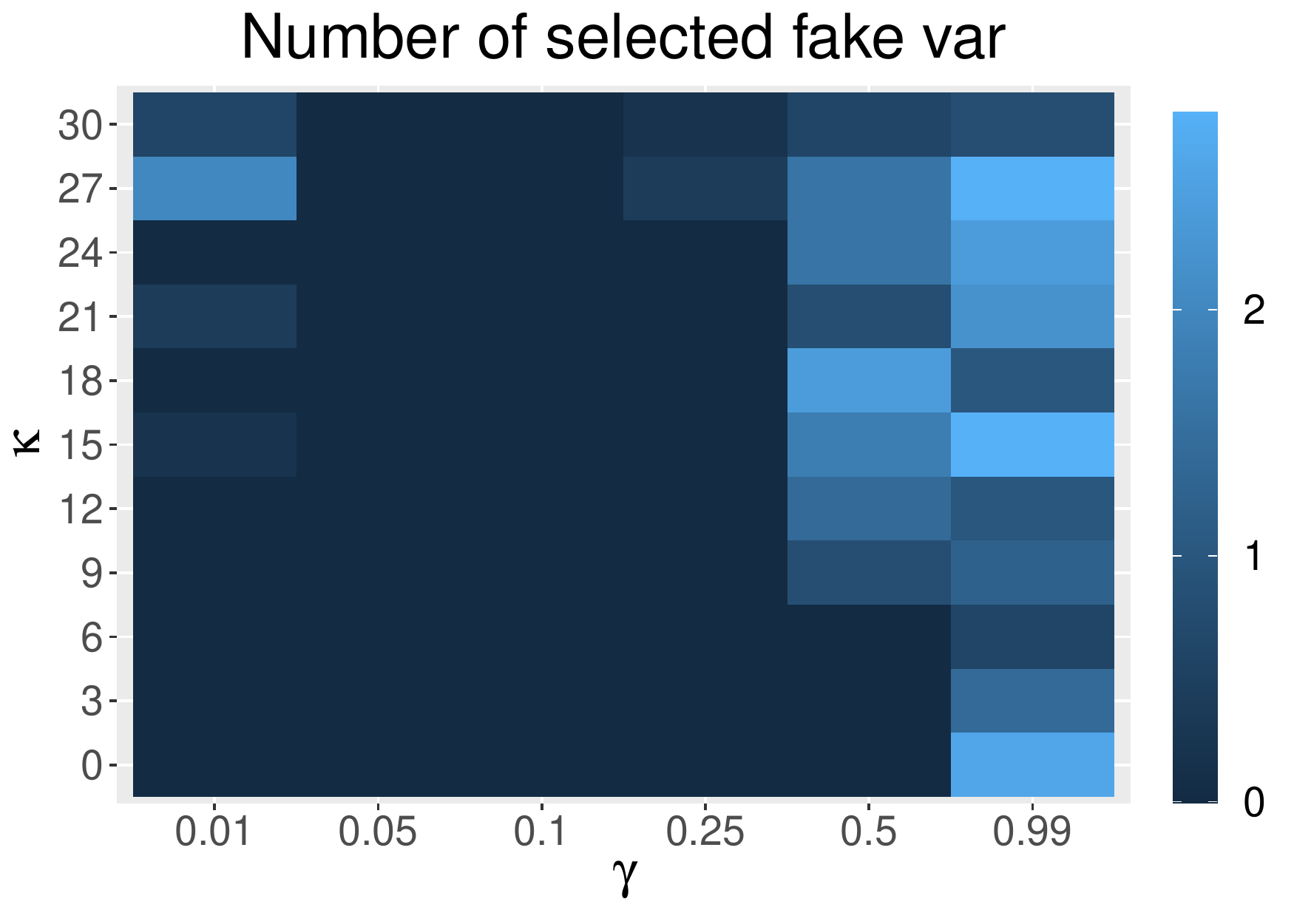}
      \caption{$k = 5$}
  \end{subfigure}
  \caption{Sensitivity of the VGPR algorithm with respect to the number of new covariates selected at each iteration ($k$), and the penalty parameters $(\gamma, \kappa)$ defined in \eqref{equ:adbridge}.}
  \label{fig:sensitivity}
\end{figure}

\section{Proofs\label{app:proofs}}

\begin{proof}[Partial proof of Proposition~\ref{thm:sqexp}]
Ignoring the constant term in $\ell$:
\begin{align*}
    \ell &= -\frac{1}{2} \by^\top \bfSigma^{-1} \by - \frac{1}{2} \log |\bfSigma| \\
    \left. \frac{\partial \ell}{\partial r_l^2} \right|_{(\sigma, \br, \tau) = (\sigma_1, \br_1, \tau_1)} &= \by^\top \bfSigma_1^{-1} \bfSigma_{l, 1} \bfSigma_1^{-1} \by - \mbox{tr}(\bfSigma_1^{-1} \bfSigma_{l, 1}),
\end{align*}
where $\bfSigma_{l, 1} = \left. \frac{\partial \bfSigma}{\partial r_l^2} \right|_{(\sigma, \br, \tau) = (\sigma_1, \br_1, \tau_1)} = - \sigma_1^2 \tilde\bfSigma_1 \odot \bD_l$, $\bD_l$ is an $n \times n$ matrix whose $(i, j)$-th coefficient is $(x_{il} - x_{jl})^2$, and $\odot$ is the Hadamard product. Because $\{x_{il}\}_{i = 1, \ldots, n, l = 1, \ldots, d}$ have i.i.d.\ distributions and $r_{l_1 1} = r_{l_2 1}$, we have:
\begin{align*}
    E[\mbox{tr}(\bfSigma_1^{-1} \bfSigma_{l_1, 1})] =
    E[\mbox{tr}(\bfSigma_1^{-1} (- \sigma_1^2 \tilde\bfSigma_1 \odot \bD_{l_1}))] =
    E[\mbox{tr}(\bfSigma_1^{-1} (- \sigma_1^2 \tilde\bfSigma_1 \odot \bD_{l_2}))] =
    E[\mbox{tr}(\bfSigma_1^{-1} \bfSigma_{l_2, 1})].
\end{align*}
Therefore, we only need to compare $E[\by^\top \bfSigma_1^{-1} \bfSigma_{l_1, 1} \bfSigma_1^{-1} \by]$ and $E[\by^\top \bfSigma_1^{-1} \bfSigma_{l_2, 1} \bfSigma_1^{-1} \by]$. First consider $\{\bx_i\}_{i = 1}^{n}$ as fixed and take expectation with respect to $\by$:
\begin{align*}
    E_{\by}[\by^\top \bfSigma_1^{-1} \bfSigma_{l,1} \bfSigma_1^{-1} \by] = \mbox{tr}(\bfSigma_1^{-1} \bfSigma_{l,1} \bfSigma_1^{-1} \bfSigma_0).
\end{align*}
$\bfSigma_0$ can be also written as:
\begin{align*}
    \bfSigma_0 = \sigma_0^2 \tilde\bfSigma_1 \odot \exp\left(-\sum_{l = d_1 + 1}^{d}(r_{l0}^2 - r_{l1}^2) \bD_{l}\right) + \tau_0^2 \bI_{n} = \sigma_0^2 \tilde\bfSigma_1 \odot \exp(-r_{l_{1}0}^2 \bD_{l_{1}} + \bC) + \tau_0^2 \bI_{n},
\end{align*}
where $\bC = - \sum_{l = \{d_1 + 1, \ldots, d_0\} \backslash l_1} r_{l0}^2 \bD_{l}$. Hence, we can re-write $E_{\by}[\by^\top \bfSigma_1^{-1} \bfSigma_{l,1} \bfSigma_1^{-1} \by]$:
\begin{align*}
    E_{\by}[\by^\top \bfSigma_1^{-1} \bfSigma_{l,1} \bfSigma_1^{-1} \by] &= \mbox{tr}\left(\bfSigma_1^{-1} (- \sigma_1^2 \tilde\bfSigma_1 \odot \bD_l) \bfSigma_1^{-1} (\sigma_0^2 \tilde\bfSigma_1 \odot \exp(-r_{l_{1}0}^2 \bD_{l_{1}} + \bC) + \tau_0^2 \bI_{n})\right) \\
    &= \mbox{tr}\left(\bfSigma_1^{-1} (- \sigma_1^2 \tilde\bfSigma_1 \odot \bD_l) \bfSigma_1^{-1} (\sigma_0^2 \tilde\bfSigma_1 \odot \exp(-r_{l_{1}0}^2 \bD_{l_{1}} + \bC))\right) + c_{\tau},
\end{align*}
where $c_{\tau} = \tau_0^2 \mbox{tr}(\bfSigma_1^{-1} (- \sigma_1^2 \tilde\bfSigma_1 \odot \bD_l) \bfSigma_1^{-1})$ and $E_{\bX}[c_{\tau}]$ remains the same for $l = l_1$ and $l = l_2$ because $\{x_{il}\}_{i = 1, \ldots, n, l = 1, \ldots, d}$ have i.i.d.\ distributions. We can also remove the $\sigma_0^2$ and $\sigma_1^2$ from the equation above with $c_{\sigma} = \sigma_0^2 \sigma_1^2$:
\begin{align}
    \label{equ:grad_expectY}
    E_{\by}[\by^\top \bfSigma_1^{-1} \bfSigma_{l,1} \bfSigma_1^{-1} \by] = c_{\sigma} \mbox{tr}\left(\bfSigma_1^{-1} (- \tilde\bfSigma_1 \odot \bD_l) \bfSigma_1^{-1} (\tilde\bfSigma_1 \odot \exp(-r_{l_{1}0}^2 \bD_{l_{1}} + \bC))\right) + c_{\tau}.
\end{align}
Since $\bfSigma_1^{-1} (- \tilde\bfSigma_1 \odot \bD_l) \bfSigma_1^{-1}$ and $\tilde\bfSigma_1 \odot \exp(-r_{l_{1}0}^2 \bD_{l_{1}} + \bC)$ are symmetric, we have:
\begin{align*}
    E_{\by}[\by^\top \bfSigma_1^{-1} \bfSigma_{l,1} \bfSigma_1^{-1} \by] &= 
    c_{\sigma} \langle \mbox{vec}\left(\bfSigma_1^{-1} (- \tilde\bfSigma_1 \odot \bD_l) \bfSigma_1^{-1}\right),
    \mbox{vec}\left(\tilde\bfSigma_1 \odot \exp(-r_{l_{1}0}^2 \bD_{l_{1}} + \bC)\right) \rangle + c_{\tau}\\
    &= c_{\sigma} \langle (\bfSigma_1^{-1} \otimes \bfSigma_1^{-1})\mbox{vec}(- \tilde\bfSigma_1 \odot \bD_l),
    \mbox{vec}\left(\tilde\bfSigma_1 \odot \exp(-r_{l_{1}0}^2 \bD_{l_{1}} + \bC) \right) \rangle + c_{\tau}\\
    &= c_{\sigma} \langle (\bfSigma_1^{-1} \otimes \bfSigma_1^{-1}) \mbox{diag}(\mbox{vec}(\tilde\bfSigma_1)) \mbox{vec}(- \bD_{l}), 
    \mbox{diag}(\mbox{vec}(\tilde\bfSigma_1)) \mbox{vec}(\exp(-r_{l_{1}0}^2 \bD_{l_{1}} + \bC)) \rangle + c_{\tau},
\end{align*}
where $\langle , \rangle$ is the Euclidean inner product, $\otimes$ is the Kronecker product, and $\mbox{vec}(\cdot)$ is the vectorization of a matrix that stacks the columns of a matrix on top of one another. Use $\bM$ to denote $\mbox{diag}(\mbox{vec}(\tilde\bfSigma_1))(\bfSigma_1^{-1} \otimes \bfSigma_1^{-1})\mbox{diag}(\mbox{vec}(\tilde\bfSigma_1))$, which is a positive definite matrix:
\begin{align*}
    E_{\by}[\by^\top \bfSigma_1^{-1} \bfSigma_{l,1} \bfSigma_1^{-1} \by] &= 
    c_{\sigma} \langle \bM \mbox{vec}(- \bD_{l}), \mbox{vec}(\exp(-r_{l_{1}0}^2 \bD_{l_{1}} + \bC))  \rangle + c_{\tau}.
\end{align*}
Now consider the expectation with respect to $\bX$ and notice that $\bM$, $\bD_{l_1}$, and $\bC$ are mutually independent. Assuming $l = l_1$ or $l = l_2$:
\begin{align}
    E_{\bX}[\langle \bM \mbox{vec}(- \bD_{l}), \mbox{vec}(\exp(-r_{l_1}^2 \bD_{l_1} + \bC)) \rangle] = \mbox{tr}\left( E_{\bX}[\mbox{vec}(- \bD_{l}) \mbox{vec}(\exp(-r_{l_1}^2 \bD_{l_1} + \bC))^\top] E_{\bX}[\bM] \right)
    \label{equ:appendix1}
\end{align}
To show \eqref{equ:appendix1} is bigger when $l = l_1$ than when $l = l_2$, it remains to show that:
\begin{align*}
    \mbox{tr}\left(\cov_{\bX}\left[\mbox{vec}(- \bD_{l_1}), \mbox{vec}\left(\exp(-r_{l_1}^2 \bD_{l_1} + \bC)\right)\right] E_{\bX}[\bM] \right) > 0.
\end{align*}
It suffices to show that $\cov_{\bX}[\mbox{vec}(- \bD_{l_1}), \mbox{vec}(\exp(-r_{l_1}^2 \bD_{l_1} + \bC))]$ is positive semi-definite and has a rank greater than zero because the trace of the multiplication between one positive definite matrix and one non-zero positive semi-definite matrix is positive. It is obvious that $\cov_{\bX}[\mbox{vec}(- \bD_{l_1}), \mbox{vec}(\exp(-r_{l_1}^2 \bD_{l_1} + \bC))]$ has a rank greater than zero. Its coefficients have three types of values:
\begin{align*}
    \cov_{\bX}[- d^{l_1}_{i_1, j_1}, \exp(-r_{l_1}^2 d^{l_1}_{i_2, j_2} + c_{i_2, j_2})] &= \cov_{\bX}[- d^{l_1}_{i_1, j_1}, \exp(-r_{l_1}^2 d^{l_1}_{i_2, j_2})] E_{\bX}[\exp(c_{i_2, j_2})] \\
    &= \left\{
    \begin{array}{cc}
        0 & i_1 = j_1 \mbox{ or } i_2 = j_2  \mbox{ or } |\{i_1, i_2, j_1, j_2\}| = 4 \\
        a & i_1 \neq j_1 \mbox{ and } \{i_1, j_1\} = \{i_2, j_2\} \\
        b & i_1 \neq j_1 \mbox{ and } i_2 \neq j_2 \mbox{ and } |\{i_1, i_2, j_1, j_2\}| = 3
    \end{array}
    \right.,
\end{align*}
where $d^{l_1}_{i, j}$ and $c_{i, j}$ denote the $(i, j)$-th coefficients of $\bD_{l_1}$ and $\bC$, respectively, and $|\cdot|$ denotes the cardinality of a set. First, we can take out $E_{\bX}[\exp(c_{i_2, j_2})]$ and have the following:
\begin{align*}
    \cov_{\bX}[\mbox{vec}(- \bD_{l_1}), \mbox{vec}(\exp(-r_{l_1}^2 \bD_{l_1} + \bC))] 
    =
    \cov_{\bX}[\mbox{vec}(- \bD_{l_1}), \mbox{vec}(\exp(-r_{l_1}^2 \bD_{l_1}))] \times c,
\end{align*}
where $c = E_{\bX}[\exp(c_{i_2, j_2})]$ for $i_2 \neq j_2$. Second, notice that the structures of 
\begin{align*}
    \cov_{\bX}[\mbox{vec}(- \bD_{l_1}), \mbox{vec}(\exp(-r_{l_1}^2 \bD_{l_1}))] \mbox{ and } \cov_{\bX}[\mbox{vec}(\bD_{l_1}), \mbox{vec}(\bD_{l_1})]
\end{align*} 
are the same, except for that the latter has different values for $a$ and $b$, denoted by $\tilde{a}$ and $\tilde{b}$, respectively. Since $\cov_{\bX}[\mbox{vec}(\bD_{l_1}), \mbox{vec}(\bD_{l_1})]$ is positive semi-definite and $a$ and $\tilde{a}$ are positive, to show that $\cov_{\bX}[\mbox{vec}(- \bD_{l_1}), \mbox{vec}(\exp(-r_{l_1}^2 \bD_{l_1}))]$ is also positive semi-definite, it is sufficient to show that $0 \le \frac{b}{a} \le \frac{\Tilde{b}}{\Tilde{a}}$. 
\begin{align}
    \frac{\Tilde{b}}{\Tilde{a}} &= \frac{\mbox{cov}[(X_1 - X_2)^2, (X_1 - X_3)^2]}{\mbox{var}[(X_1 - X_2)^2]}
    \label{equ:thres},\\
    \frac{b}{a} &= \frac{\mbox{cov}[-(X_1 - X_2)^2, \exp(- r^2(X_1 - X_3)^2)]}{\mbox{cov}[-(X_1 - X_2)^2, \exp(- r^2(X_1 - X_2)^2)]}. \label{equ:ratio}
\end{align}
Without loss of generality, we can assume $X_1, X_2, X_3 \overset{i.i.d.}{\sim} U(-\frac{1}{2}, \frac{1}{2})$ or $X_1, X_2, X_3 \overset{i.i.d.}{\sim} N(0, 1)$, under which $\frac{b}{a} \le \frac{\Tilde{b}}{\Tilde{a}}$ can be shown numerically as in Figure~\ref{fig:numeical_ratios}.
\begin{figure}[H]
\centering
	\begin{subfigure}{.49\textwidth}
	\centering
 	\includegraphics[width =.8\linewidth]{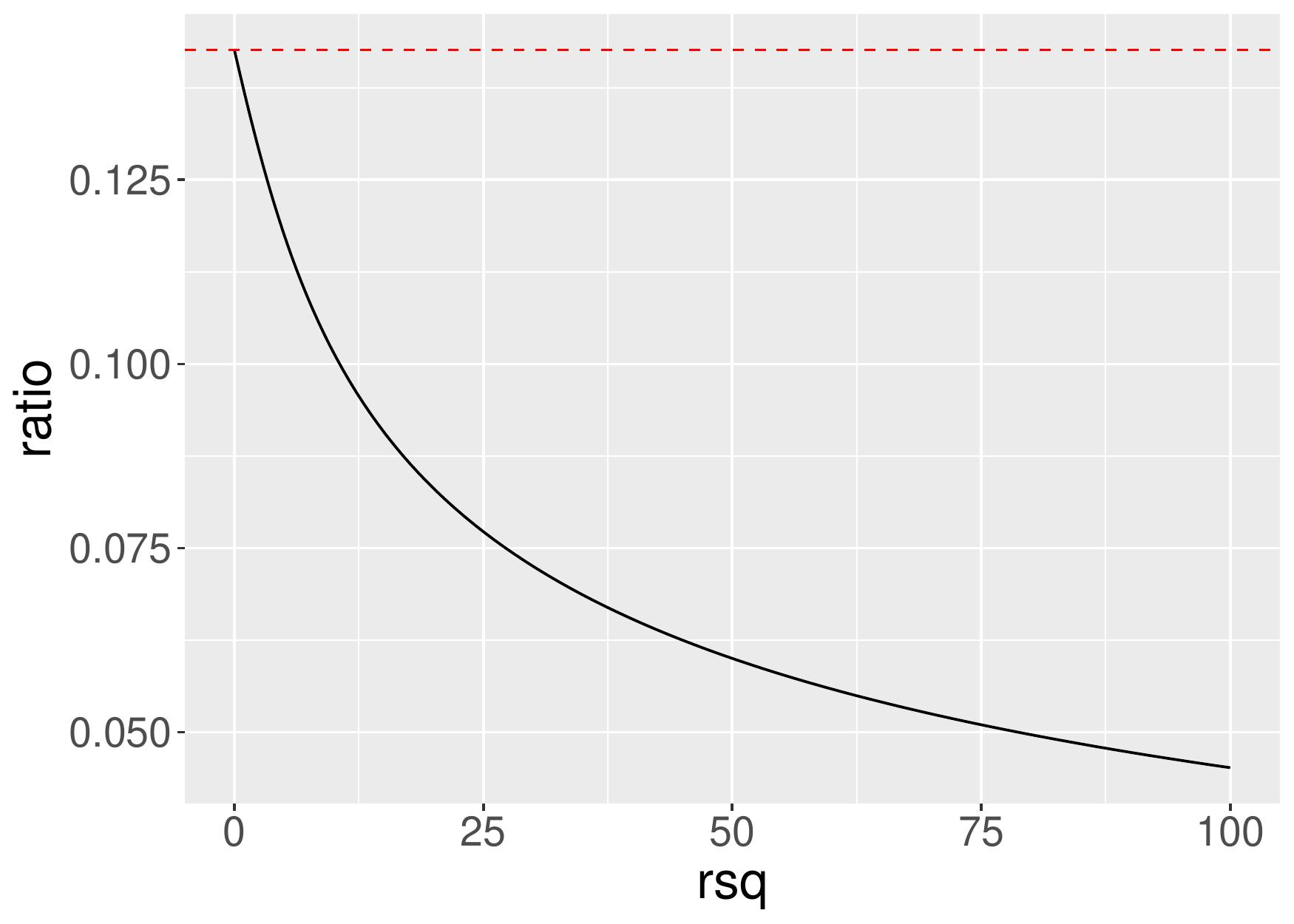}
	\caption{}
	\end{subfigure}%
\hfill
	\begin{subfigure}{.49\textwidth}
	\centering
 	\includegraphics[width =.8\linewidth]{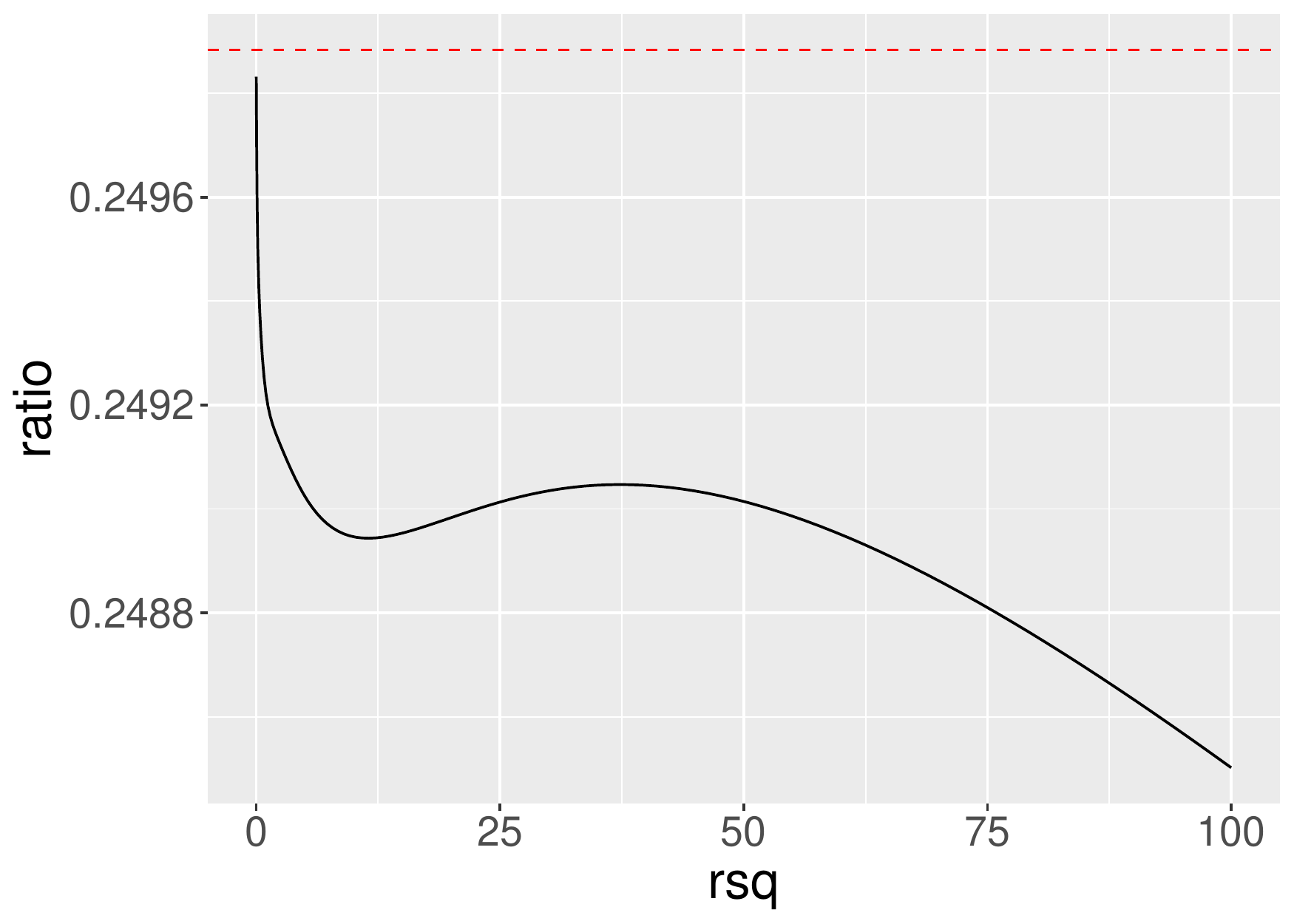}
	\caption{}
	\end{subfigure}
	
	\caption{The ratio \eqref{equ:ratio} (solid black) and the threshold \eqref{equ:thres} (dashed red) under (a) uniform distribution and (b) normal distribution for $X_i$.}
	\label{fig:numeical_ratios}
\end{figure}
With the accurately computed relationship between \eqref{equ:thres} and \eqref{equ:ratio}, we conclude the proof under the assumption that $\{x_{il}\}_{i = 1, \ldots, n, l = 1, \ldots, d}$ have i.i.d. uniform or normal distributions. For other distributions, similar numerical procedure can be used to draw the conclusion.
\end{proof}

\begin{proof}[Proof of Corollary~\ref{cor:sqexp_at_zero}]
When $\br \rightarrow \bfzero$, $\bfSigma \rightarrow \mathbf{1}_{n \times n}$. We can substitute $\tilde\bfSigma_1$ in the proof of Proposition~\ref{thm:sqexp} by a non-singular covariance matrix arbitrarily close to $\mathbf{1}_{n \times n}$, denoted by $\tilde{\bfSigma}_1$, substitute $d_1$ by $0$, and hence 
$$
\bfSigma_0 \approx \sigma_0^2 \tilde\bfSigma_1 \odot \exp\left(-\sum_{l = d_1 + 1}^{d}(r_{l0}^2 - r_{l1}^2) \bD_{l}\right) + \tau_0^2 \bI_{n}.
$$ 
The rest of the proof should remain the same.
\end{proof}

\begin{proof}[Proof of Corollary~\ref{cor:sqexp_corr}]
Based on \eqref{equ:appendix1}, 
\begin{align*}
    &E\left[\left. \frac{\partial \ell}{\partial r_{l_1}^2} \right|_{(\sigma, \br, \tau) = (\sigma_1, \br_1, \tau_1)}\right] - E\left[\left. \frac{\partial \ell}{\partial r_{d + 1}^2} \right|_{(\sigma, \br, \tau) = (\sigma_1, \br_1, \tau_1)}\right]\\
    ={} &c_{\sigma}E_{\bX}[\langle \bM \mbox{vec}(- \bD_{l_1}), \mbox{vec}(\exp(-r_{l_1}^2 \bD_{l_1} + \bC)) \rangle] - c_{\sigma}E_{\bX}[\langle \bM \mbox{vec}(- \bD_{d + 1}), \mbox{vec}(\exp(-r_{l_1}^2 \bD_{l_1} + \bC)) \rangle] \\
    &\mbox{Because $\{x_{il}\}_{i = 1, \ldots, n, l = 1, \ldots, d}$ have i.i.d.\ normal distributions, } E_{\bX}[\mbox{vec}(- \bD_{d + 1})] = E_{\bX}[\mbox{vec}(- \bD_{l_1})],\\
    ={} &c_{\sigma} \rho_2^2 \mbox{tr}\left(\cov_{\bX}[\mbox{vec}(- \bD_{l_1}), \mbox{vec}(\exp(-r_{l_1}^2 \bD_{l_1} + \bC))] E_{\bX}[\bM] \right) > 0.
\end{align*}
\end{proof}

\begin{proof}[Proof of Proposition~\ref{thm:all_kernel}]
From the proof of Proposition~\ref{thm:sqexp}, we know that
\begin{align*}
    &E\left[\left. \frac{\partial \ell}{\partial r_{l_1}^2} \right|_{(\sigma, \br, \tau) = (\sigma_1, \br_1, \tau_1)}\right] - E\left[\left. \frac{\partial \ell}{\partial r_{l_2}^2} \right|_{(\sigma, \br, \tau) = (\sigma_1, \br_1, \tau_1)}\right] \\
    ={} &E_{\bX}[E_{\by}[\by^\top \bfSigma_1^{-1} \bfSigma_{l_1,1} \bfSigma_1^{-1} \by]] - 
    E_{\bX}[E_{\by}[\by^\top \bfSigma_1^{-1} \bfSigma_{l_2,1} \bfSigma_1^{-1} \by]].\\
\end{align*}
Based on \eqref{equ:grad_expectY}, 
\begin{align*}
    & E_{\by}[\by^\top \bfSigma_1^{-1} \bfSigma_{l_1,1} \bfSigma_1^{-1} \by] - E_{\by}[\by^\top \bfSigma_1^{-1} \bfSigma_{l_2,1} \bfSigma_1^{-1} \by]\\
    ={} & c_{\sigma} \mbox{tr}\left(\bfSigma_1^{-1} (- \tilde\bfSigma_1 \odot \bD_{l_1}) \bfSigma_1^{-1} \tilde\bfSigma_0\right) - 
    c_{\sigma} \mbox{tr}\left(\bfSigma_1^{-1} (- \tilde\bfSigma_1 \odot \bD_{l_2}) \bfSigma_1^{-1} \tilde\bfSigma_0\right)\\
    ={} &c_{\sigma} \mbox{tr}\left(\bfSigma_1^{-1} (- \tilde\bfSigma_1 \odot (\bD_{l_1} - \bD_{l_2})) \bfSigma_1^{-1} \tilde\bfSigma_0\right) \\
    \approx{} & c_{\sigma} \mbox{tr}\left(\bfSigma_1^{-1} (- \tilde\bfSigma_1 \odot (\bD_{l_1} - \bD_{l_2})) \bfSigma_1^{-1} \left(\tilde\bfSigma_1 + \sum_{\tilde{l} = d_1 + 1}^{d_0} \tilde\bfSigma_{\tilde l, 1} r_{\tilde{l}0}^2\right)\right)\\
    ={} & c_{\sigma} \mbox{tr}\left(\bfSigma_1^{-1} (- \tilde\bfSigma_1 \odot (\bD_{l_1} - \bD_{l_2})) \bfSigma_1^{-1} \left(\tilde\bfSigma_1 - \sum_{\tilde{l} = d_1 + 1}^{d_0} r_{\tilde{l}0}^2 \tilde\bfSigma_1 \odot \bD_{\tilde l} \right)\right),\\
\end{align*}
where $\tilde \bfSigma_{l, 1} = \left. \frac{\partial \tilde\bfSigma}{\partial r_l^2} \right|_{(\sigma, \br, \tau) = (1, \br_1, 0)} = - \tilde\bfSigma_1 \odot \bD_l$. Noticing that $E_{\bX}[\bD_{l_1} - \bD_{l_2}]$ is a zero matrix, that $\{\bD_l\}_{l = 1}^{d}$ are mutually independent, and that $\bD_{l_1}$ and $\bD_{l_2}$ are independent from $\tilde \bfSigma_{1}$, the expectation of the above equation with respect to $\bX$ is equal to:
\begin{align*}
    &c_{\sigma} E_{\bX}\left[ \mbox{tr}\left(\bfSigma_1^{-1} (- \tilde\bfSigma_1 \odot (\bD_{l_1} - \bD_{l_2})) \bfSigma_1^{-1} \left( - r_{l_{10}}^2 \tilde\bfSigma_1 \odot \bD_{l_1} \right)\right) \right] \\
    ={} &c_{\sigma} r_{l_{10}}^2 E_{\bX}\left[ \mbox{tr}\left( (\bfSigma_1^{-1} \otimes \bfSigma_1^{-1})  \mbox{vec}(\tilde\bfSigma_1 \odot (\bD_{l_1} - \bD_{l_2})) \mbox{vec}\left( \tilde\bfSigma_1 \odot \bD_{l_1} \right)^\top\right) \right] \\
    ={} &c_{\sigma} r_{l_{10}}^2 \mbox{tr}\left( E_{\bX}[\bfSigma_1^{-1} \otimes \bfSigma_1^{-1}]  \cov_{\bX}\left[\mbox{vec}(\tilde\bfSigma_1 \odot (\bD_{l_1} - \bD_{l_2})), \mbox{vec}\left( \tilde\bfSigma_1 \odot \bD_{l_1} \right) \right] \right)  \\
    ={} &c_{\sigma} r_{l_{10}}^2 \mbox{tr}\left( E_{\bX}[\bfSigma_1^{-1} \otimes \bfSigma_1^{-1}]  \cov_{\bX}\left[\mbox{vec}(\tilde\bfSigma_1 \odot \bD_{l_1}), \mbox{vec}( \tilde\bfSigma_1 \odot \bD_{l_1} ) \right] \right) > 0,\\
\end{align*}
because $E_{\bX}[\bfSigma_1^{-1} \otimes \bfSigma_1^{-1}]$ is positive-definite and $\cov_{\bX}\left[\mbox{vec}(\tilde\bfSigma_1 \odot \bD_{l_1}), \mbox{vec}( \tilde\bfSigma_1 \odot \bD_{l_1} ) \right]$ is positive semi-definite with a rank greater than zero.
\end{proof}

\begin{proof}[Proof of Proposition~\ref{thm:Vecchia_unbiased}]
    Here, we take expectations only with respect to $\by$ and consider $\{\bx_i\}_{i = 1}^{n}$ as fixed. Using the non-negativeness of the KL divergence, we can show that for a generic random vector $\bw$, whose distribution is parameterized by $\bftheta_0$:
    \begin{align*}
        E[\log p(\bw; \bftheta_0)] - E[\log p(\bw; \bftheta)] &= \int \log \frac{p(\bw; \bftheta_0)}{p(\bw; \bftheta)} p(\bw; \bftheta_0) \mbox{d} \bw \ge 0 \\
        \Rightarrow E[\log p(\bw; \bftheta_0)] &\ge E[\log p(\bw; \bftheta)],
    \end{align*}
    with which Proposition~\ref{thm:Vecchia_unbiased} can be thus proved:
    \begin{align*}
        E_{\by}[\hat{\ell}^{\tilde{\br}}(\bftheta)] &= E_{\by}[\log \hat{\dens}_{\bftheta}^{\tilde{\br}}(\by)] = E_{\by}[\log \prod_{i = 1}^{n} p_{\bftheta}(y_i | \by_{c(i)})] = E_{\by}[\sum_{i = 1}^{n} \log p_{\bftheta}(y_i | \by_{c(i)})] \\
        &= \sum_{i = 1}^{n} E_{\by_{c(i)}}\left[E_{y_i | \by_{c(i)}}[\log p_{\bftheta}(y_i | \by_{c(i)})]\right],
    \end{align*}
    where $E_{y_i | \by_{c(i)}}[\log p_{\bftheta}(y_i | \by_{c(i)})]$ achieves maximum at $\bftheta = \bftheta_0$. Therefore, $E_{\by}[\hat{\ell}_{\bftheta}^{\tilde{\br}}(\by)]$ achieves maximum at $\bftheta = \bftheta_0$.
\end{proof}

\end{appendices}
\bibliographystyle{apalike}
\bibliography{mendeley,additionalrefs}

\end{document}